\documentclass[sn-mathphys-num]{sn-jnl}

\usepackage{graphicx}%
\usepackage{multirow}%
\usepackage{amsmath,amssymb,amsfonts}%
\usepackage{amsthm}%
\usepackage{mathrsfs}%
\usepackage[title]{appendix}%
\usepackage{xcolor}%
\usepackage{textcomp}%
\usepackage{manyfoot}%
\usepackage{booktabs}%
\usepackage{algpseudocode}%
\usepackage{listings}%
\usepackage{natbib}
\usepackage[utf8]{inputenc}
\usepackage{booktabs}
\usepackage{enumitem}
\usepackage{threeparttable}
\usepackage[justification=centering]{caption}
\usepackage[ruled]{algorithm2e}
\usepackage{float}

\theoremstyle{thmstyleone}%
\newtheorem{theorem}{Theorem}
\theoremstyle{thmstyletwo}%
\newtheorem{example}{Example}%
\newtheorem{remark}{Remark}%

\theoremstyle{thmstylethree}%
\newtheorem{definition}{Definition}%
\newtheorem{lemma}{Lemma}

\begin{document}

\title[Article Title]{On the construction of ultra-light MDS matrices\footnote{This research was supported by National Key Research and Development Project under Grant No. 2018YFA0704705 and CAS Project for Young Scientists in Basic Research (Grant No. YSBR-035).}}


\author[1]{\fnm{Yu} \sur{Tian}}

\author*[2]{\fnm{Xiutao} \sur{Feng}}\email{fengxt@amss.ac.cn}

\author[1]{\fnm{Guangrong} \sur{Li}}

\affil[1]{\orgdiv{College of Artificial Intelligence}, \orgname{Nanning Vocational and Technical University} \orgaddress{\street{} \city{Naning}, \postcode{} \state{} \country{China}}}

\affil*[2]{\orgdiv{Key Laboratory of Mathematics Mechanization, Academy of Mathematics and Systems Science}, \orgname{Chinese Academy of Sciences},
\orgaddress{\street{} \city{Beijing}, \postcode{} \state{} \country{China}}}


\abstract{In recent years, the Substitution-Permutation Network has emerged as a crucial structure for constructing symmetric key ciphers. Composed primarily of linear matrices and nonlinear S-boxes, it offers a robust foundation for cryptographic security. Among the various metrics used to assess the cryptographic properties of linear matrices, the branch number stands out as a particularly important index. Matrices with an optimal branch number are referred to as MDS matrices and are highly prized in the field of cryptography. In this paper we delve into the construction of lightweight MDS matrices. We commence implementation trees of MDS matrices, which is a vital tool for understanding and manipulating their implementations, and then present an algorithm that efficiently enumerates all the lightest MDS matrices based on the word representation. As results, we obtain a series of ultra-lightweight $4\times 4$ MDS matrices, remarkably, 4-bit input MDS matrices with 35 XOR operations and 8-bit input ones with 67 XOR operations . These matrices represent the most comprehensive lightweight MDS matrices available to date. Furthermore, we craft some involution $4\times 4$ MDS matrices with a mere 68 XOR gates.To our best knowledge, they are the best up to date. In the realm of higher-order MDS matrices, we have successfully constructed $5\times 5$ and $6\times 6$ matrices with 114 and 148 XOR gates respectively. These findings outperform the current state-of-the-art.}




\maketitle

\section{Introduction}\label{sec1}

The Substitution-Permutation Network (SPN, for brevity) has emerged as a pivotal component in the construction of symmetric ciphers, particularly in recent years. Due to the prevalence of the block cipher AES \cite{daemen2013design}, this structure is extensively studied and typically comprises three primary elements: a key schedule, a small-sized (usually 4- or 8-bits) nonlinear function known as the S-box, and a larger (often 32-, 64- or 128-bits) linear function. The S-box serves to mix the bits within a 4- or 8-bit word, while the linear function mixes words together.

According to the wide trail design strategy \cite{daemen2001wide}, the resilience of SPN ciphers against classical attacks, particularly differential and linear attacks, can be assessed by examining their constituent components. Specifically, for a nonlinear S-box, a high nonlinearity and low differential uniformity are desirable; S-boxes achieving optimal differential uniformity are termed Almost Perfect Nonlinear (APN). As for the linear function, a high branch number is expected to establish strong diffusion between input and output words; matrices corresponding to linear functions with an optimal branch number are known as MDS matrices (related to Maximum Distance Separable codes). Besides their use in SPN ciphers like AES, MDS matrices also find application in Feistel ciphers (e.g., Camellia \cite{aoki2001128}, SMS4 \cite{StateCryptographyAdministration}), hash functions (e.g., Whirlpool \cite{barreto2000whirlpool}, Grostl \cite{gauravaram2009grostl}), and stream ciphers (e.g., ZUC \cite{feng2011zuc}).

With the increasing miniaturization of electronic devices handling sensitive data, there is a growing need for novel cryptographic primitives that offer low implementation costs. In contrast to standardized, robust primitives like AES, lighter cryptographic solutions are often preferred in resource-constrained environments. To reduce the cost of the SPN structure, one approach is to optimize its main components: the S-box and the linear function. Numerous studies have explored lightweight alternatives for these components, as evidenced by works on S-boxes \cite{li2014constructing, canteaut2015construction} and linear functions \cite{sajadieh2012recursive, wu2012recursive, augot2013exhaustive, berger2013construction, augot2014direct, sim2015lightweight, li2016construction, liu2016lightweight, sarkar2016lightweight}. These findings have facilitated the design of various new cipher proposals aimed at balancing cost and security in constrained settings, such as Present \cite{bogdanov2007present}, KATAN \cite{de2009katan}, LED \cite{guo2011led}, LBlock \cite{wu2011lblock}, Prince \cite{borghoff2012prince}, Skinny \cite{beierle2016skinny}, among others. These examples underscore the advantages of low-cost cryptographic primitives.

MDS matrices, with their optimal branch number and excellent diffusion properties, are well-suited for linear diffusion layers and are widely used in cryptographic designs. However, it's important to note that MDS matrices can be relatively dense, resulting in higher hardware implementation costs(larger circuit areas). Therefore, to incorporate MDS matrices as linear diffusion layers in lightweight ciphers, it becomes necessary to optimize existing MDS matrices or construct new lightweight MDS matrices that strike a balance between security and efficiency.

\subsection{Related works}

There has been a significant amount of research dedicated to constructing lightweight MDS matrices or near-MDS matrices.

One common approach for constructing MDS matrices involves selecting a matrix with a specific structure, such as Hadamard, Cauchy, Vandermonde, Toeplitz, circulant, or other matrices with special properties. Then, specific values are assigned to the variable parameters of the matrix to ensure that all its submatrices are nonsingular, resulting in an MDS matrix. This method reduces the search space by focusing on one or a few specific types of matrices. Additionally, for matrices with special structures, there are often many repeated values in the determinants of their submatrices, which reduces the computational complexity when verifying the nonsingularity of the submatrices. Remarkably, for Cauchy matrices, their structure guarantees that they are MDS matrices regardless of the parameter values. In summary, searching for MDS matrices within matrices with special structures is computationally efficient and often successful. There is a considerable amount of literature on this topic, including works based on Cauchy matrices \cite{cui2010construction}, \cite{cui2014compact}, \cite{dau2015constructions}, \cite{cui2021construction}, \cite{mousavi2021involutory},  Vandermonde matrices  \cite{mattoussi2012complexity}, \cite{li2018direct}, \cite{yaici2019particular}, \cite{yu2020comparison}, Hadamard matrices \cite{tan2017orthogonal}, \cite{pehlivanouglu2018generalisation}, \cite{cong2020new}, \cite{you2021construction}, \cite{pehlivanouglu2021construction}, \cite{zhou2022construction},  Toeplitz matrices \cite{sarkar2017analysis}, \cite{sharma2017constructions},  \cite{pehlivanoglu2018generating}, \cite{sakalli2020lightweight}, \cite{chen2019constructions}, and circulant matrices \cite{cauchois2019circulant}, \cite{malakhov2021construction}, \cite{wang2022more},  \cite{adhiguna2022orthogonal}, \cite{wang2022inverse}.

However, this approach of limiting the search to specific matrix structures has some drawbacks. The implementation cost of the resulting matrices is often difficult to control. Early methods focused on optimizing the implementation of matrix elements, known as local optimization. For example, in a circulant matrix, where the entire matrix is determined by its first row consisting of $k$ elements (which are $n\times n$ invertible matrices over $F_2$), local optimization involves optimizing the implementation of these $k$ matrices while preserving the XOR operations between them. This approach cannot avoid the overhead of $(k-1)kn$ XOR operations for the entire matrix. Relevant works include \cite{jean2017optimizing}, \cite{zhou2018efficient}, \cite{kolsch2019xor}. Specifically, in \cite{kolsch2019xor}, K\"{o}lsch provided the specific structure of all matrices that can be implemented using only 2 XORs for element (treated as a matrix) multiplication with a vector in a field of characteristic 2. Local optimization methods have a smaller computational footprint because they only focus on optimizing matrix elements, but their effectiveness is limited as they do not consider the overall matrix structure.

To improve optimization efficiency, a natural approach is to treat a $k\times k$ matrix with  $n\times n$ invertible matrices as elements as a single $kn\times kn$ matrix over $F_2$ and optimize the entire matrix, known as global optimization \cite{kranz2017shorter}. It involves in the so-called SLP (Straight Line Program) problem. In \cite{boyar2013logic}, Boyar et al. reduced the point set covering problem to a subproblem of the SLP problem (SLP problem with Hamming weight limited to 3 for each row) and proved that the SLP problem is NP-hard. This indicates that finding a universally applicable and efficient algorithm for it is challenging. Currently, heuristic algorithms are commonly used to address this issue, including the Paar algorithm \cite{paar1997optimized}, the BP algorithm \cite{boyar2008shortest} and its variants \cite{visconti2018improved}, \cite{reyhani2018smashing}, \cite{maximov2019new}, \cite{banik2019more}, \cite{T2020heuristics}, \cite{baksi2021three}.

Recently, Yang et al. introduced a new search method for involutory MDS matrices and obtained a $4\times 4$ involutory MDS matrix implementable with only 35 XORs (4-bit input) and 70 XORs (8-bit input)~\cite{yang2021construction}. Here it should be pointed out that their 35-XOR result has being optimal in the 4-bit input. Pehlivano?lu et al. further considered 3-input and higher XOR gates in their variant of the BP algorithm and achieved MDS matrices with optimized circuit area \cite{pehlivanouglu2023construction}.

An alternative approach to constructing lightweight MDS matrices is to limit the implementation cost of the matrix, then derive the matrix from its implementation, and finally check if it satisfies the MDS condition, i.e., all submatrices are nonsingular. This method was first introduced by Duval and Leurent in \cite{duval2018mds} and extended by Zeng et al \cite{wang2023four, ShiWang2022}. Roughly speaking,  it considers matrices of the form $M = A_1A_2\cdots A_t$ with an upper bound $t$, where $A_1,A_2,\cdots,A_t$ have low implementation costs. By means of careful selection for $A_i$, some $4\times 4$ MDS matrix with 35 XOR gates (4-bit input) and 67 XOR gates (8-bit input) are constructed, which have been proven to be optimal under a word-based structure \cite{venkateswarlu2022lower}. In \cite{sajadieh2021construction}, Sajadieh and Mousavi further applied it to search higher-order MDS matrices with a General Feistel Structure and obtained some  $6\times 6$ and $8\times 8$ MDS matrices with relatively low costs.

\subsection{Our contributions}

In this study, our main objective is to construct the lightest weight MDS matrices. We first investigate the structures of MDS matrices with low implementation costs and give all possible structures with optimal costs. Based on these structures, we then generate a significant number of new ultra-light MDS matrices. 

More speaking, we introduce the concept of an "implementation tree". Based on the implementation tree, we develop a new method that starts with several unit vectors and iteratively generates rows for the MDS matrix. As each new row is generated, we check if it satisfies the MDS condition. This verification step is crucial as it eliminates a large number of unproductive branches during the search process. Through this refined approach, we are able to identify all MDS matrix structures with optimal costs (based on word metrics). By leveraging these structures, we successfully find all 60 nonequivalent $4\times 4$ MDS matrices that can be implemented with just 67 XOR operations on the ring $F_2[x]/(x^8 + x^2 + 1)$ and 35 XOR operations on the ring $F_2[x]/(x^4 + x + 1)$. Notably, our results encompass the 10 matrices in \cite{wang2023four} and the 52 matrices in \cite{ShiWang2022}. Here it should be pointed that our method is also capable of producing a significant quantity of $4\times 4$ MDS matrices with 36-41 XOR gates for 4-bit input and 68-80 XOR gates for 8-bit input by modestly relaxing the implementation cost constraints. Due to the limit of the paper, we only list the counts of those with the lowest cost for various depths in Table \ref{tab:2}. Furthermore, we identify $4\times 4$ involutory MDS matrices on the same ring that require only 68 XOR operations for implementation. To the best of our knowledge, it is the most optimal result currently available.

To demonstrate the effectiveness of our algorithm, we conduct a comparative analysis with the methodology described in \cite{duval2018mds}. The results of this comparison are presented in Table \ref{comparison}.

\begin{center}
\scalebox{0.7}{
\begin{threeparttable}[htbp]
	
	\caption{The number of MDS matrices of different depths}
	
		\label{tab:2} 
		\begin{tabular}{cccccc}
			\hline\noalign{\smallskip}
			Ring & Depth & Cost & \cite{duval2018mds} & \cite{ShiWang2022} & Our results  \\
			\noalign{\smallskip}\hline\noalign{\smallskip}
			$F_2[\alpha]$ & -  & 67 & 2 & 52 & 60\\
			& 5 & 67 & 2 & 4 & 4\\
			& 4 & 69 & 1 & 4 & 7\\
			& 3 & 77 & 1 & 2 & 3 \\
			$F_2[\beta]$ & -  & 35 & 2 & 52 & 60\\
			& 5 & 35 & 2 & 4 & 4\\
			& 4 & 37 & 1 & 4 & 7\\
			& 3 & 41 & 1 & 2 & 3 \\
			\noalign{\smallskip}\hline
		\end{tabular}
		\begin{tablenotes}
			\item  $\alpha$ is the companion matrix of $x^8+x^2+1$, $\beta$ is the companion matrix of $x^4+x+1$.
		\end{tablenotes}
	
\end{threeparttable}
}
\end{center}

\begin{center}
\begin{threeparttable}[htbp]
	\caption{\label{comparison}Comparison of the efficiency of our algorithm with previous results}
	\centering
	\begin{tabular}{cccccccc}
		\hline
		\multirow{2}{*}{Ring} & \multirow{2}{*}{Cost} & \multirow{2}{*}{Depth} & \multicolumn{2}{c}{Previous Algorithm} & \multicolumn{2}{c}{Our Algorithm} & \multirow{2}{*}{Ref.} \\
		\cline{4-5}\cline{6-7}
		& & & Memory & Time & Memory & Time & \\
		\hline
		$F_2[\alpha]$ & 67 & 5 & 30.9G & 19.5h & \textbf{4M} & $<$\textbf{10ms} & \cite{duval2018mds}\\
		$F_2[\alpha]$ & 68 & 5 & 24.3G & 2.3h & \textbf{4M} & $<$\textbf{10ms} & \cite{duval2018mds}\\
		$F_2[\alpha]$ & 69 & 4 & 274G & 30.2h & \textbf{4M} & $<$\textbf{10ms} & \cite{duval2018mds}\\
		\hline
	\end{tabular}
	\begin{tablenotes}
		\item \qquad $\alpha$ is the companion matrix of $x^8+x^2+1$.
	\end{tablenotes}
\end{threeparttable}
\end{center}

As for higher-order MDS matrices, we successfully constructs $5\times 5$ and $6\times 6$ matrices for 8-bit inputs with 114 and 148 XOR gates respectively, which are the best results up to date, see Table \ref{tab:3}.

\begin{center}
\begin{threeparttable}[H]

	\caption{The XOR gates required to implement MDS matrices of different orders}
	
		\label{tab:3}       
		\begin{tabular}{cccc}
			\hline\noalign{\smallskip}
			Size of matrix & Previous results & Our results & Ref. \\
			\noalign{\smallskip}\hline\noalign{\smallskip}
			5 & 129  & 114 & \cite{Ke2019Exhaustive} \\
			6 & 156 & 148 & \cite{sajadieh2021construction} \\
			\noalign{\smallskip}\hline
		\end{tabular}
	
\end{threeparttable}
\end{center}

\subsection{Organization of the paper}

\par In Section 2, we introduce several essential conceptions that serve as the foundation for our subsequent discussions. In Section 3, we propose a novel algorithm aimed at efficiently searching for ultra-light $4\times 4$ MDS (Maximum Distance Separable) matrices and involutory MDS matrices. In Section 4, we extend our focus to the construction of higher-order MDS matrices, especially for $5\times 5$ and $6\times 6$ MDS matrices. Finally, in Section 5 we summarize our key findings and conclusions.

\section{Preliminaries}

Since linear functions defined over vector spaces can be equivalently expressed as matrix-vector multiplications, we restrict our attention to the corresponding matrix representations in the following discussion. Let $n$ and $k$ be positive integers, and let $F_2$ denote the binary field consisting of the elements 0 and 1.

In the context of SPN ciphers, linear functions often operate on the concatenated outputs of multiple S-boxes. Here, we specifically consider a $k\times k$ matrix, whose entries are $n\times n$ invertible matrices over $F_2$. In this setting, $k$ represents the number of S-boxes, while n represents the size of a single S-box. Alternatively, this structure can be interpreted as a larger $nk\times nk$ matrix defined over $F_2$. For simplicity and consistency, we further constrain these $n\times n$ matrices to belong to a commutative ring $R_n$ that is a subset of $GL(n,2)$ - the general linear group of $n\times n$ invertible matrices over the binary field $F_2$. This restriction applies throughout the remainder of this paper. As a concrete example, one might consider the case where $R_n=F_2[\alpha]$ and $\alpha$ is the companion matrix of a polynomial of degree $n$. We denote the set of all $k\times k$ invertible matrices with entries in $R_n$ by $M_k(R_n)$.

\subsection{MDS matrix}

The branch number is a significant cryptographic parameter used to measure the diffusion properties of linear functions. Given that any arbitrary linear function $L$ can be represented in matrix form, i.e., there exists a matrix $M$ such that $L(x) = Mx$, the concept of the branch number can be extended to any matrix. In the following, we introduce the notion of the branch number for matrices.

\begin{definition}
	Let $M \in M_k(R_n)$ be a given matrix. Its differential and linear branch numbers are defined respectively as:	
	$$
	B_d(M) = \min_{x \neq 0}\{\text{wt}(x) + \text{wt}(Mx)\} \text{ and } B_l(M) = \min_{x \neq 0}\{\text{wt}(x) + \text{wt}(M^Tx)\},
	$$
	where $\text{wt}(x)$ denotes the weight of $x$, i.e., the number of nonzero components in $x$, and $M^T$ represents the transpose of $M$.
\end{definition}

It is evident that for any matrix $M$ in $M_k(R_n)$, its differential and linear branch numbers are at most $k+1$. If either of these branch numbers equals $k+1$, the matrix is designated as an MDS matrix. Notably, if the differential branch number of a matrix achieves optimality, its linear branch number does so as well, and vice versa.

For a given matrix $M$, the following theorem provides a necessary and sufficient condition for $M$ to be an MDS matrix.

\begin{theorem}[\cite{Au2013Exhaustivesearch}]
	Let $R_n$ be a commutative ring and $M$ be a $k \times k$ matrix over $R_n$. Then $M$ is an MDS matrix if and only if all of its minors are invertible.
\end{theorem}

Let $\mathbf{e}_i$ denote the $i$-th unit vector, that is, $\mathbf{e}_i = (0,\ldots,0,\stackrel{i}{1},0,\ldots,0)$ with $1$ in the $i$-th position. A \textbf{permutation matrix} is an invertible matrix whose rows consist of unit vectors.

\begin{theorem}
	Let $M$ be an MDS matrix and $P, Q$ be two permutation matrices. Then $M' = PMQ$ is also an MDS matrix.
\end{theorem}

For any given MDS matrix $M$, an MDS matrix $M'$ is said to be \emph{equivalent} to $M$ if there exist two permutation matrices $P$ and $Q$ such that $M' = PMQ$, denoted by $M' \sim M$. Since the inverse of a permutation matrix is also a permutation matrix, this equivalence relation is symmetric, meaning if $M' \sim M$, then $M \sim M'$. Let $[M]$ represent the set of all MDS matrices equivalent to $M$. It is straightforward to observe that $[M]$ comprises $(k!)^2$ MDS matrices.

\subsection{Cost and depth of binary matrices}

The differential and linear branch numbers constitute two pivotal cryptographic indicators of a matrix. When a matrix realized as a circuit, its depth and cost emerge as two additional essential parameters. Subsequently, we study the circuit realization of a general binary matrix, emphasizing its depth and cost implications.

\begin{definition}
	Consider any nonzero binary $m\times n$ matrix $M$. An implementation $I$ of $M$ refers to a circuit constructed solely using XOR gates that can compute $y=Mx$ for any input vector $x$, where $x = (x_1,x_2,\ldots,x_n)\in F_2^n$ and $y = (y_1,y_2,\ldots,y_m)\in F_2^m$.
\end{definition}

As an illustrative example, let us consider a $4\times4$ binary matrix
\begin{equation*}
	M=\begin{bmatrix}
		1 & 0 & 0 & 0\\
		1 & 1 & 0 & 0\\
		1 & 1 & 1 & 0\\
		1 & 1 & 1 & 1
	\end{bmatrix}.
\end{equation*}
Its implementation $I$ can be expressed as:
\begin{align*}
	\begin{cases}
		y_1 &= x_1, \\
		y_2 &= x_1 \oplus x_2, \\
		y_3 &= (x_1 \oplus x_2) \oplus x_3, \\
		y_4 &= (x_1 \oplus x_2) \oplus (x_3 \oplus x_4).
	\end{cases}
\end{align*}

Figure \ref{fig:ibm} depicts the structure of $I$, where intermediate terms $x_5 = x_1 \oplus x_2$ and $x_7 = x_3 \oplus x_4$ are introduced for clarity. Since $I$ comprises solely of XOR gates, it gives rise to a binary tree $B_i$ rooted at $y_i$ for each $y_i$ ($1 \leq i \leq m$), known as the spanning tree of $y_i$. In $B_i$, nodes other than the root $y_i$ are denoted as descendants of $y_i$.

\begin{center}
\begin{figure}[htbp]
	\includegraphics[height=4cm, width=12cm]{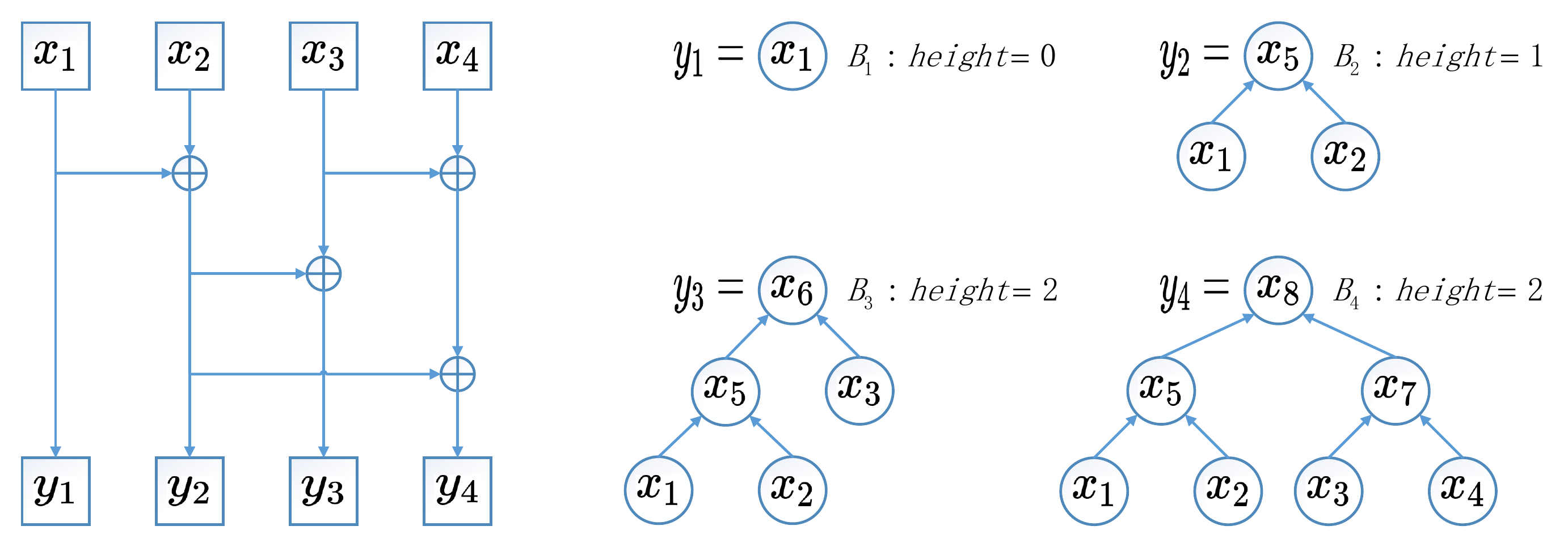}
	\caption{An implementation $I$ of $M$}
	\label{fig:ibm}
\end{figure}
\end{center}

\begin{definition}
	For a given binary matrix $M$ with an implementation $I$, the cost of $I$ is defined as the total count of XOR gates employed in $I$.
\end{definition}

\begin{definition}
	Given a binary $m\times n$ matrix $M$ and its implementation $I$, the depth of $I$ is determined as the maximum height attained among all the trees $B_i$ rooted at $y_i$ ($1 \leq i \leq m$).
\end{definition}

Considering the aforementioned $4\times4$ binary matrix $M$ and its implementation $I$, we observe that the depth and cost of $I$ are 2 and 4, respectively. Notably, a binary matrix $M$ may admit multiple distinct implementations, each potentially exhibiting different costs and depths. For instance, Figure \ref{fig:ibm2} presents an alternative implementation $I'$ of the same $4\times4$ matrix $M$, with a depth and cost of 3. In our work, we prioritize implementations with the lowest cost for any given matrix $M$.

\begin{center}
\begin{figure}[htbp]
	\includegraphics[height=4.5cm, width=12cm]{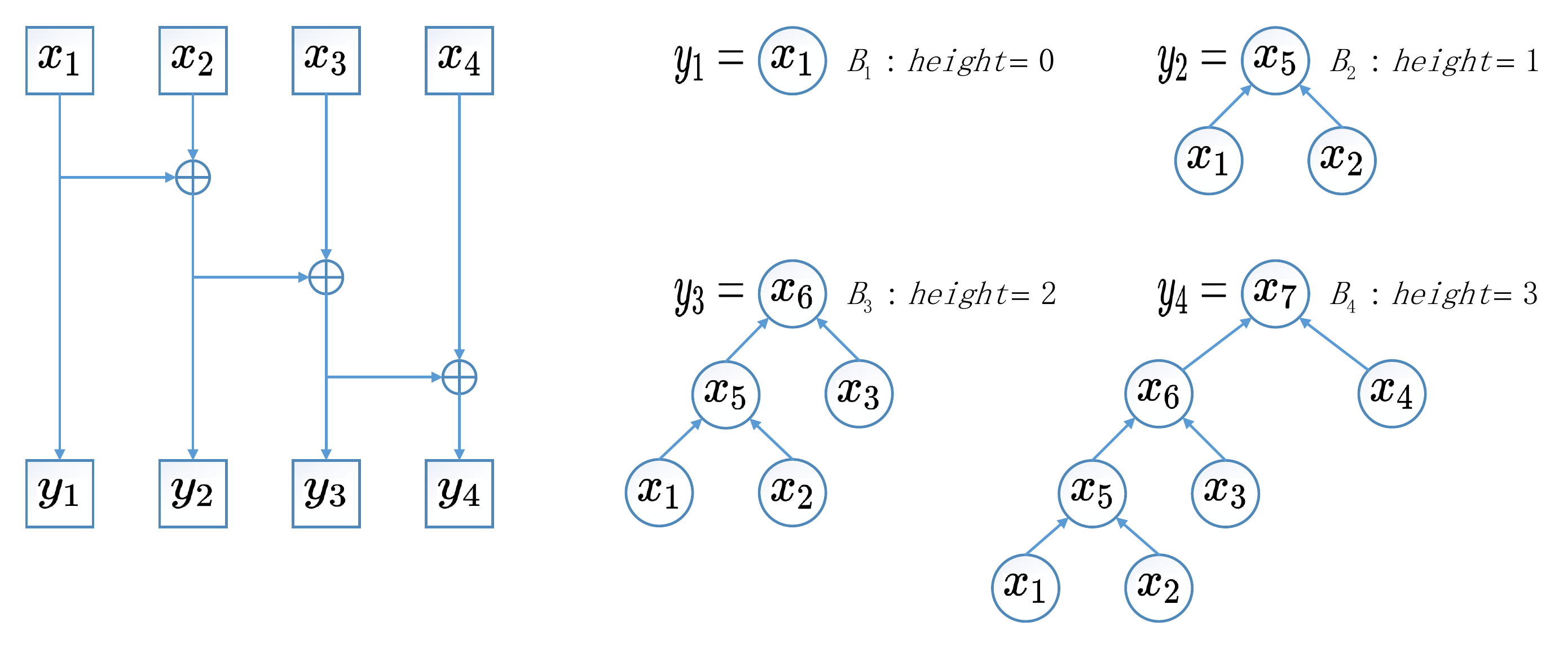}
	\caption{Another implementation $I'$ of $M$}
	\label{fig:ibm2}
\end{figure}
\end{center}

Let $M$ and $M'$ be matrices related by $M' = PMQ$, where $P$ and $Q$ are permutation matrices, and let $I$ be an implementation of $M$. An implementation $I'$ of $M'$ is said to be a corresponding implementation of $I$ if it is obtained solely by rearranging the input and output variables in accordance with $P$ and $Q$. Thus we have the following conclusion.

\begin{theorem}\label{theorem:3}
	Given equivalent matrices $M$ and $M'$, with $I$ being an implementation of $M$ and $I'$ being the corresponding implementation of $I$ for $M'$, it follows that $I$ and $I'$ have identical cost and depth.
\end{theorem}

By Theorem \ref{theorem:3}, it becomes evident that all matrices belonging to the same equivalence class share a common cost and depth. Therefore, in our search for ultra-light MDS matrices, we can confine our attention to the distinct equivalence classes.

\subsection{Metric}

In this section, we discuss the method to compute the cost of a $k \times k$ matrix $M$ over $R_n$:
\begin{equation*}
	\begin{bmatrix}
		y_{1} \\
		y_{2} \\
		\vdots \\
		y_{k}
	\end{bmatrix}
	=
	\begin{bmatrix}
		a_{11} & a_{12} & \cdots & a_{1k} \\
		a_{21} & a_{22} & \cdots & a_{2k} \\
		\vdots & \vdots & \ddots & \vdots \\
		a_{k1} & a_{k2} & \cdots & a_{kk}
	\end{bmatrix}
	\begin{bmatrix}
		x_{1} \\
		x_{2} \\
		\vdots \\
		x_{k}
	\end{bmatrix},
\end{equation*}
where $x_i, y_i \in F_2^n$, $a_{ij} \in R_n$, $1 \leq i, j \leq k$. Here, we only consider block implementations of $M$, treating each $x_i$ and $y_i$ as individual blocks.

For a given matrix $M$ and its block implementation $I$, two fundamental operations exist:
\begin{enumerate}
	\item Word-wise XOR: $X \gets X_1 \oplus X_2$,
	\item Scalar multiplication over $R_n$: $X \gets aX$,
\end{enumerate}

where $X, X_1, X_2$ are $n$-dimensional vectors over $F_2$ and $a \in R_n$. Word-wise XOR requires $n$ XOR gates; the cost of scalar multiplication $aX$ over $R_n$ depends on the specific value of $a$.

Any block implementation $I$ of matrix $M$ can be expressed as a linear straight-line program, as exemplified below.

\begin{example}\label{exg:1}
	\begin{equation*}
		\begin{bmatrix}
			y_1 \\
			y_2 \\
			y_3 \\
			y_4
		\end{bmatrix}
		=
		\begin{bmatrix}
			1 & \alpha & 1 & 0 \\
			1 & 0 & 1 & 1 \\
			0 & \alpha & 1 & 1 \\
			1 & 1 & 1 & 0
		\end{bmatrix}
		\begin{bmatrix}
			x_1 \\
			x_2 \\
			x_3 \\
			x_4
		\end{bmatrix},
	\end{equation*}
	where $x_i, y_i \in F_2^8$ $(1 \leq i \leq 4)$, and $\alpha$ is the companion matrix of $x^8 + x^2 + 1$:
	
	\begin{small}
		\begin{equation*}
			\alpha =
			\begin{bmatrix}
				0 & 0 & 0 & 0 & 0 & 0 & 0 & 1 \\
				1 & 0 & 0 & 0 & 0 & 0 & 0 & 0 \\
				0 & 1 & 0 & 0 & 0 & 0 & 0 & 1 \\
				0 & 0 & 1 & 0 & 0 & 0 & 0 & 0 \\
				0 & 0 & 0 & 1 & 0 & 0 & 0 & 0 \\
				0 & 0 & 0 & 0 & 1 & 0 & 0 & 0 \\
				0 & 0 & 0 & 0 & 0 & 1 & 0 & 0 \\
				0 & 0 & 0 & 0 & 0 & 0 & 1 & 0 \\
			\end{bmatrix}.
		\end{equation*}
	\end{small}
	
One implementation of $M$ is given by:
\begin{align*}
	\begin{cases}
		y_1 &= (x_1 \oplus x_3) \oplus \alpha x_2, \\
		y_2 &= (x_1 \oplus x_3) \oplus x_4, \\
		y_3 &= (x_3 \oplus x_4) \oplus \alpha x_2, \\
		y_4 &= (x_1 \oplus x_3) \oplus x_2.
	\end{cases}
\end{align*}
This can be rearranged into a linear straight-line program $I$ as follows:
\begin{align*}
	x_5 &= x_1 \oplus x_3, \\
	x_6 &= x_3 \oplus x_4, \\
	x_7 &= x_5 \oplus \alpha x_2 = y_1, \\
	x_8 &= x_5 \oplus x_4 = y_2, \\
	x_9 &= x_6 \oplus \alpha x_2 = y_3, \\
	x_{10} &= x_5 \oplus x_2 = y_4.
\end{align*}

Note that each step of $I$ has the form:
$$
X = a \cdot X_1 \oplus b \cdot X_2,
$$
where $X_1, X_2$ are either intermediate results $x_i(i\ge 5)$ calculated previously or input variables $x_i$ $(1 \leq i \leq 4)$.

Two distinct types of operations are observed in this form: word-wise XOR, denoted as \textbf{word-xor}(expressed as \textbf{XOR}) to distinguish it from bit-wise xor(expressed as \textbf{xor}), and scalar multiplication, which refers to matrix-vector multiplication over $R_n$.

To compute the cost of $I$, we proceed as follows:

First, the linear straight-line program comprises 6 steps, yielding an XOR cost of $6 \times 8 = 48$.

Second, considering scalar multiplication, the non-trivial multiplication in $I$ is $\alpha x_2$, which appears twice. Since the latter occurrence of $\alpha x_2$ can reuse the result of the first, only one calculation is needed when assessing cost. Therefore, the cost of scalar multiplication is 1.

In conclusion, the total cost of $I$ is $48 + 1 = 49$ xor operations.
\end{example}

From the above example, it is evident that when determining the cost of an implementation $I$ of a $k \times k$ matrix $M$ over $R_n$, we follow a two-step process. Firstly, we ascertain the number of steps $s$ in $I$ as a linear straight-line program, dictating the number of XORs necessary in $I$. Secondly, we identify all unique and non-trivial scalar multiplications in $I$ and tally the total number of xors $t$ required for these scalar multiplications. Ultimately, the cost of $I$ is given by $ns + t$.

\section{The construction of lightweight $4\times4$ MDS matrices}

In this section, we mainly focus on the construction of the $4\times 4$ MDS matrix. Because of its small size, we can discuss it in more details. First, we will introduce the concept of implementation tree.

\subsection{Implementation tree}

In Example \ref{exg:1}, we observe that under the word-based structure, a matrix implementation can be conceptualized as a linear straight-line program. Each such program uniquely determines a matrix. The lightweight MDS matrix can be constructed by searching for the linear straight-line program with the lowest cost capable of implementing an MDS matrix. Next, we investigate the properties of these linear straight-line programs.

Consider a linear straight-line program realizing a matrix:
\[ t_1 = a_1 t_{1,1} \oplus b_1 t_{1,2}, t_2 = a_2 t_{2,1} \oplus b_2 t_{2,2}, \ldots, t_m = a_m t_{m,1} \oplus b_m t_{m,2}, \]
where $t_{i,j}$ ($1 \leq i \leq m, 1 \leq j \leq 2$) are either initial variables $x_p$ or intermediate results $t_q$ ($1 \leq q < i$). We designate $t_q$ ($1 \leq q \leq m$) as a \textbf{term} of the program, and the number of terms $m$ as its \textbf{length}.

A partial order relation can be established among the terms $t_1, \ldots, t_m$: if $t_p = a_pt_q \oplus b_pt_{p,2}$ or $t_p = a_pt_{p,1} \oplus b_pt_q$ appears in the program, then $t_q \prec t_p$. Furthermore, if $t_q \prec t_p$ and $t_r \prec t_q$, then $t_r \prec t_p$. For instance, if $t_5 = at_3 + bt_4$ and $t_4 = ct_1 + dt_2$, then $t_i \prec t_5$ for $1 \leq i \leq 4$. The relation $t_q \prec t_p$ signifies that $t_q$ is an essential component of $t_p$.

Given the partial order relationship, suppose the following linear straight-line program can implement some $k \times k$ matrix on $R_n$, with $a_pt_{p,1} \oplus b_pt_{p,2} = t_p$ where $a_p \neq 0$ and $b_p \neq 0$:

\begin{align*}
	a_1t_{1,1} \oplus b_1t_{1,2} &= t_1, \\
	&\vdots \\
	a_{i_1}t_{i_1,1} \oplus b_{i_1}t_{i_1,2} &= t_{i_1} = y_1, \\
	a_{i_1+1}t_{i_1+1,1} \oplus b_{i_1+1}t_{i_1+1,2} &= t_{i_1+1}, \\
	&\vdots \\
	a_{i_2}t_{i_2,1} \oplus b_{i_2}t_{i_2,2} &= t_{i_2} = y_2, \\
	&\vdots \\
	a_{i_k}t_{i_k,1} \oplus b_{i_k}t_{i_k,2} &= t_{i_k} = y_k.
\end{align*}

This program exhibits the following properties:

\begin{enumerate}
	\item If $t_p$ is not an output (i.e., $t_p \neq y_i$ for all $1 \leq i \leq k$), then $t_p \prec y_i$ for some $1 \leq i \leq k$.
	\item If $1 \leq q < p \leq i_k$, then either $t_q \prec t_p$ or $t_q$ and $t_p$ are incomparable. If $t_q$ and $t_{p}$ are incomparable, their order in the program can be swapped.
\end{enumerate}

Let us elaborate on these properties:

\begin{enumerate}
	\item The partial order relation $t_b \prec t_a$ indicates that $t_b$ is necessary for computing $t_a$, that is to say, $t_a$ cannot be computed before $t_b$. Conversely, if $t_b \not\prec t_a$, whether $t_b$ has been generated is irrelevant to the computation of $t_a$. For Property 1, if there exists a non-output $t_p$ such that $t_p \not\prec y_i$ for all $1 \leq i \leq k$, this implies that $t_p$ is not required for computing any outputs $y_i$ ($1 \leq i \leq k$). Therefore, removing $t_p$ from the matrix implementation would not affect its functionality, and we can reduce the program's length and simplify the implementation. Thus we always assume that no such a $t_p$ exists in a matrix implementation.
	
	\item Consider $t_q$ and $t_p$ ($1 \leq q < p \leq i_k$). If $t_p \prec t_q$, it means that $t_p$ must be used to compute $t_q$. However, since $t_p$ is generated after $t_q$, this leads to a contradiction. Therefore, their relationship can only be $t_q \prec t_p$ or they are incomparable. If $t_p$ and $t_{p+1}$ are incomparable, swapping their positions in the program has no impact on the overall implementation. This is because there are no terms between $t_p$ and $t_{p+1}$, so the swap does not affect the generation of these two terms or any subsequent terms. However, it should be noted that if $t_p$ and $t_{p+1}$ are incomparable and $p - q \geq 2$, such a swap may not be feasible. For example, if there exists $q < s < p$ such that $t_s \prec t_p$, swapping $t_q$ and $t_p$ would place $t_s$ after $t_p$, preventing the generation of $t_p$ and leading to a contradiction.
\end{enumerate}

For brevity, we represent the above linear straight-line program as $[t_1, \ldots, t_{i_1} = y_1, \ldots, t_{i_k} = y_k]$. Applying properties 1 and 2, we arrive at the following theorem:

\begin{theorem}\label{th:3-1}
	For any given $k \times k$ MDS matrix $M$ on $R_n$, an implementation $I = [t_1, \ldots, t_{i_1} = y_1, \ldots, t_{i_k} = y_k]$ of $M$ can be rearranged to $I' = [t_1', \ldots, t_{i_1'}' = y_1, \ldots, t_{i_k}' = y_k]$ by swapping the relative positions of $t_p$ ($1 \leq p \leq i_k$) such that $I'$ satisfies:
	\begin{align*}
		t_1' &\prec y_1, \ldots, t_{i_1'-1}' \prec y_1, \\
		t_{i_1'+1}' &\prec y_2, \ldots, t_{i_2'-1}' \prec y_2, \\
		&\vdots \\
		t_{i_{k-1}'+1}' &\prec y_k, \ldots, t_{i_k-1}' \prec y_k.
	\end{align*}
\end{theorem}  
\begin{proof}
	Consider the implementation $I$ of $M$, which we divide into $k$ segments as follows:
	\begin{align*}
		t_1,&\cdots,t_{i_1-1},t_{i_1}=y_1,\\
		t_{i_1+1},&\cdots,t_{i_2-1},t_{i_2}=y_2,\\
		&\vdots\\
		t_{i_{k-1}+1}&\cdots,t_{i_k-1},t_{i_k}=y_k.
	\end{align*}
	
	Utilizing Property 2, we commence the adjustment process segment by segment, starting from the first.
	
	In the first segment, we initiate the downward search from $t_{i_1-1}$. Suppose the first term encountered that is incomparable with $y_1$ is $t_{j_1}$. Then, under the partial order relation $\prec$, $t_{j_1}$ is incomparable with $t_{j_1+1},\cdots,t_{i_1}=y_1$. This arises because $t_{j_1}$ is the inaugural term discovered in the descending search that cannot be compared with $y_1$. If there existed some $p$ such that $t_{j_1}\prec t_p$ ($j_1 < p < i_1$) held true, then since $t_p\prec y_1$, we would have $t_{j_1}\prec y_1$?a contradiction. Hence, $t_{j_1}$ is indeed incomparable with $t_{j_1+1},\cdots,t_{i_1}=y_1$. By Property 2, we can successively exchange $t_{j_1}$ upward, akin to bubble sorting, until it is positioned behind $y_1$. This process is repeated until the remaining terms $t_1',\cdots,t_{i_1'-1}'$ satisfy $t_s\prec y_1$ for all $1\le s < i_1'$. Thus, the first segment is appropriately adjusted.
	
	Segments $2$ to $k-1$ are adjusted analogously to the first segment. After the preceding adjustments, any term $t_p$ in segment $k$ necessarily satisfies $t_p\not\prec y_s$ for all $1\le s\le k-1$. By virtue of Property 1, $t_p\prec y_k$ ensures that the last segment naturally conforms to the requirements subsequent to the adjustments of the preceding $k-1$ segments.
	
	After the aforementioned adjustments, $I'$ fulfills the theorem's conditions, thereby the conclusion follows.
\end{proof}

Below we designate $I'$ as a \textbf{normal implementation}. Notably, our adjustments solely alter the order of terms within the linear straight-line program. Therefore they are identical when they are realized as circuits.

To illustrate Theorem \ref{th:3-1}, we employ the matrix and implementation from Example \ref{exg:1}. In this instance, the implementation $I$ takes the following form:
\begin{align*}
	x_1\oplus x_3 &=x_5,\\
	x_3\oplus x_4 &=x_6,\\
	x_5\oplus \alpha x_2 &=x_7=y_1,\\
	x_5\oplus x_4 &=x_8=y_2,\\
	x_6\oplus \alpha x_2&=x_9=y_3,\\
	x_5\oplus x_2 &=x_{10}=y_4.
\end{align*}
The partial order relationships among the terms are as follows:
$$
x_5\prec x_7,x_5\prec x_8,x_5\prec x_{10},x_6\prec x_9.
$$
Adjusting $I$ according to Theorem \ref{th:3-1}: Since $x_6$ is incomparable with $x_7=y_1$, we exchange $x_6$ and $x_7$ and finish the first segment's adjustment. Next, as $x_6$ is incomparable with $x_8$, we swap $x_6$ and $x_8$ in the second segment's adjustment. For the third segment, the relationship $x_6\prec x_9=y_3$ already exists, we do nothing. Similarly, the fourth segment requires no modification.

After renumbering the terms, we obtain $I'$ as follows:
\begin{align*}
	x_1\oplus x_3 &=x_5,\\
	x_5\oplus \alpha x_2 &=x_6=y_1,\\
	x_5\oplus x_4 &=x_7=y_2,\\
	x_3\oplus x_4 &=x_8,\\
	x_8\oplus \alpha x_2&=x_9=y_3,\\
	x_5\oplus x_2 &=x_{10}=y_4.
\end{align*}

The distinction between $I$ and $I'$ lies solely in the ordering of terms, as both ultimately express the same set of relationships:
\begin{align*}
	\begin{cases}
		y_1=(x_1\oplus x_3)\oplus \alpha x_2,\\
		y_2=(x_1\oplus x_3)\oplus x_4,\\
		y_3=(x_3\oplus x_4)\oplus \alpha x_2,\\
		y_4=(x_1\oplus x_3)\oplus x_2.
	\end{cases}
\end{align*}

Based on the preceding discussion, when an MDS matrix is implemented, it suffices to consider a normal implementation akin to the form $I'$ in Theorem \ref{th:3-1}. We call the vector $(i_1', i_2' - i_1', \ldots, i_k' - i_{k-1}')$ as a characteristic of $I'$. Once the order among $y_i$ is fixed, its characteristic remains constant.

To further illustrate the concept, let us consider the MixColumn matrix of the AES algorithm as an example. We introduce the notion of an implementation tree through this example.

\begin{example}\label{exg:2}
	The MixColumn matrix of AES is given by:
	
	\[
	M_{AES} =
	\left[ \begin{array}{cccc}
		\alpha & \alpha+1 & 1 & 1 \\
		1 & \alpha & \alpha+1 & 1 \\
		1 & 1 & \alpha & \alpha+1 \\
		\alpha+1 & 1 & 1 & \alpha
	\end{array} \right],
	\]
	where $\alpha$ denotes the companion matrix of $x^8 + x^4 + x^3 + x + 1$.
	
	An implementation $I$ of $M_{AES}$ can be expressed as follows:
	
	\begin{threeparttable}[H]
		\centering
		\resizebox{130mm}{!}{ 
			\begin{tabular}{cccc}
				$x_5 = 1 \cdot x_1 \oplus 1 \cdot x_2$, & $x_6 = 1 \cdot x_2 \oplus \alpha \cdot x_5$, & $x_7 = 1 \cdot x_3 \oplus 1 \cdot x_4$, & $x_8 = 1 \cdot x_6 \oplus 1 \cdot x_7=y_1$, \\
				$x_9 = 1 \cdot x_2 \oplus 1 \cdot x_3$, & $x_{10} = 1 \cdot x_3 \oplus \alpha \cdot x_9$, & $x_{11} = 1 \cdot x_1 \oplus 1 \cdot x_4$, & $x_{12} = 1 \cdot x_{10} \oplus 1 \cdot x_{11}=y_2$, \\
				$x_{13} = 1 \cdot x_4 \oplus \alpha \cdot x_7$, & $x_{14} = 1 \cdot x_5 \oplus 1 \cdot x_{13}=y_3$, & & \\
				 $x_{15} = 1 \cdot x_1 \oplus \alpha \cdot x_{11}$, & $x_{16} = 1 \cdot x_9 \oplus 1 \cdot x_{15}=y_4$. & & \\
		\end{tabular}}
	\end{threeparttable}
	
	The cost associated with $I$ is $12 \times 8 + 4 \times 3 = 108$. Notably, the cost of XOR operations is 96 and comprises the majority of the total implementation cost. If we abstract away the specific values of scalar multiplications, the implementation $I$ can be generalized as follows:
	
	\begin{threeparttable}[H]
		\centering
		\resizebox{130mm}{!}{
			\begin{tabular}{cccc}
				$x_5 = a_1 \cdot x_1 \oplus a_2 \cdot x_2$, & $x_6 = a_3 \cdot x_2 \oplus a_4 \cdot x_5$, & $x_7 = a_5 \cdot x_3 \oplus a_6 \cdot x_4$, & $x_8 = a_7 \cdot x_6 \oplus a_8 \cdot x_7=y_1$, \\
				$x_9 = a_9 \cdot x_2 \oplus a_{10} \cdot x_3$, & $x_{10} = a_{11} \cdot x_3 \oplus a_{12} \cdot x_9$, & $x_{11} = a_{13} \cdot x_1 \oplus a_{14} \cdot x_4$, & $x_{12} = a_{15} \cdot x_{10} \oplus a_{16} \cdot x_{11}=y_2$, \\
				$x_{13} = a_{17} \cdot x_4 \oplus a_{18} \cdot x_7$, & $x_{14} = a_{19} \cdot x_5 \oplus a_{20} \cdot x_{13}=y_3$, & & \\
				$x_{15} = a_{21} \cdot x_1 \oplus a_{22} \cdot x_{11}$, & $x_{16} = a_{23} \cdot x_9 \oplus a_{24} \cdot x_{15}=y_4$. & & \\
		\end{tabular}}
	\end{threeparttable}
	
	This abstraction highlights the role of XOR operations in the implementation. It can be visually represented using a binary tree, as shown in Figure \ref{fig:3-1}.
	\begin{figure}[h]
		\centering
		\includegraphics[height=5.2cm, width=12cm]{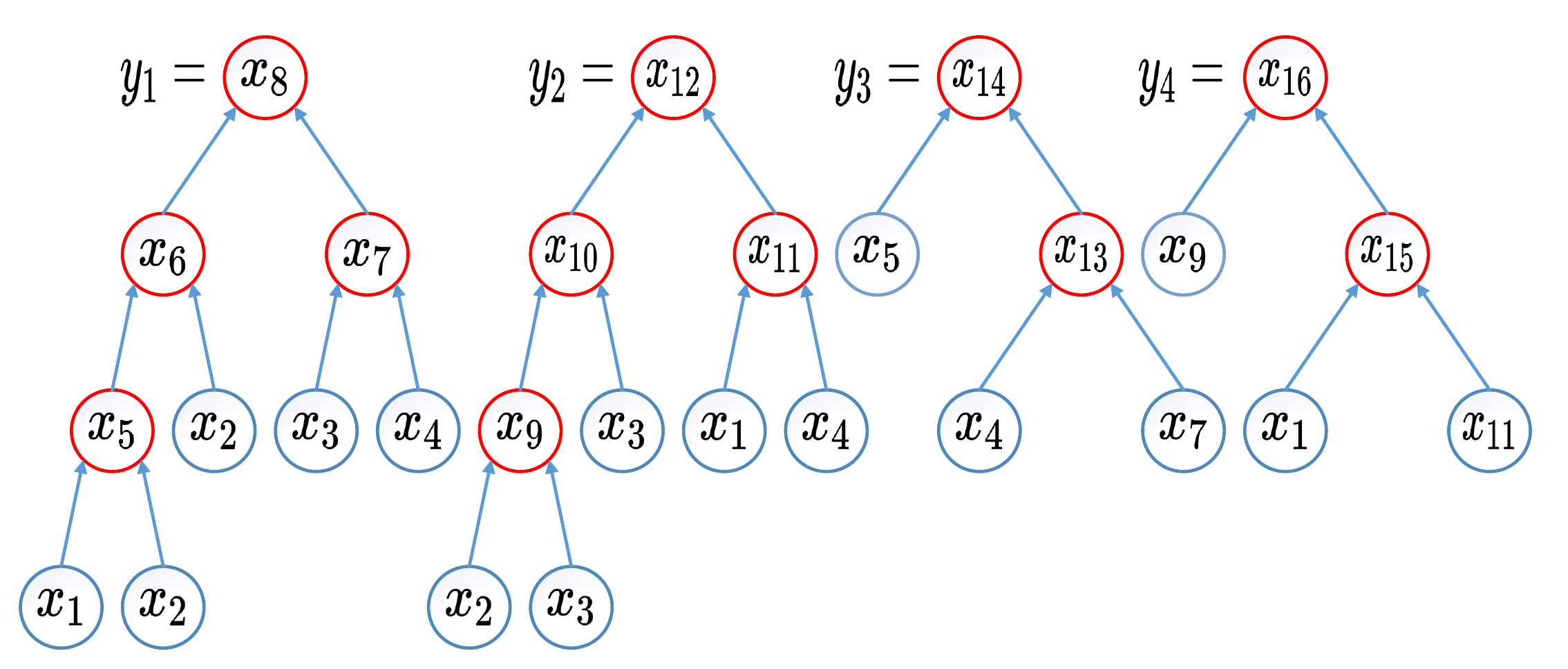}
		\caption{The implementation tree of the MixColumn matrix of AES}
		\label{fig:3-1}
	\end{figure}
	
\end{example}

From Figure \ref{fig:3-1}, we can observe the following:
\begin{enumerate}
	\item The tree is regular, meaning each node is either a leaf node or has two non-empty subtrees.
	\item Each node in the tree is generated from two nodes with smaller serial numbers, defining the order of implementation.
\end{enumerate}

\begin{definition}
	Consider a given $k \times k$ MDS matrix $M$ on $R_n$. Let $I = [t_1, \ldots, t_{i_1} = y_1, \ldots, t_{i_k} = y_k]$ be a normal implementation of $M$. We stipulate that $x_1 = t_0, x_2 = t_{-1}, \ldots, x_k = t_{-(k-1)}$. For any term $t_p = a_p t_m \oplus b_p t_n$ in $I$, it induces a triple $T_p = (m, n, p)$. If $t_s = y_i$, we mark its corresponding triple as $T_s \rightarrow y_i$. By arranging such triples in the order of their appearance in the original implementation $I$, we obtain a sequence $T = [T_1, \ldots, T_{i_1} \rightarrow y_1, \ldots, T_{i_k} \rightarrow y_k]$.
	
	We refer to $T$ as the \textbf{implementation tree} of $M$. Each triple $T_p (1 \leq p \leq i_k)$ is called a \textbf{node} of $T$. The vector $(i_1, i_2 - i_1, \ldots, i_k - i_{k-1})$ is designated as the \textbf{type} of $T$, and $i_k$ is termed the \textbf{capacity} of $T$. Among all $k \times k$ MDS matrices on $R_n$, the implementation tree with the minimum capacity is deemed the \textbf{simplest tree}.
	\label{de:3-1}
\end{definition}

In Definition \ref{de:3-1}, the implementation tree is represented as triples in the form $T_p = (m,n,p)$. Alternatively, it can be expressed using XOR operations with parameters. Let $x_1=T_0, x_2=T_{-1}, \ldots, x_k=T_{-(k-1)}$ and the implementation tree be denoted as $T=[T_1, \ldots, y_1=T_{i_1}, \ldots, y_k=T_{i_k}]$, where $T_p = a_pT_m \oplus b_pT_n$ corresponds to $T_p=(m,n,p)$. Here, $a_p$ and $b_p$ are formal parameters without specific values. Notably, the partial order relationship between terms in $I$ can be seamlessly translated to the nodes in $T$.

We define the implementation tree as follows: matrix $\rightarrow$ matrix implementation $\rightarrow$ implementation tree of matrix. Initially, we have the matrix, from which we derive the word-based implementation of the matrix, and finally extract the concept of the implementation tree. When constructing a matrix, we follow the reverse process: first, we construct an implementation tree, then determine the parameter values for each scalar multiplication to achieve a concrete implementation, and finally derive the matrix from this implementation.

Before presenting the specific algorithm, we estimate the capacity of the MDS matrix implementation tree.

\begin{lemma}\label{le:3-1}
	For a given $k \times k$ MDS matrix $M$ on $R_n$, let $I=[t_1, \ldots, t_{i_1}=y_1, \ldots, t_{i_k}=y_k]$ be an implementation of $M$. When fully expanding $t_{i_p}=y_p$ ($1 \leq p \leq k$) to the input variable $x_q$ ($1 \leq q \leq k$), each $x_q$ must appear.
\end{lemma}

\begin{proof}
	Each $t_{i_p}=y_p$ ($1 \leq p \leq k$) corresponds to row $p$ of $M$. Since $M$ is an MDS matrix, the elements in row $p$ are non-zero. Therefore, each $x_q$ ($1 \leq q \leq k$) must appear in the expansion of $t_{i_p}=y_p$.
\end{proof}

The following theorem provides a lower bound for the capacity of MDS matrix implementations.

\begin{theorem}\label{th:3-2}
	For a given $k \times k$ ($k \geq 3$) MDS matrix $M$ on $R_n$, let $T=[T_1, \ldots, T_{i_1} \rightarrow y_1, \ldots, T_{i_k} \rightarrow y_k]$ be an implementation tree of $M$. Then, $i_k \geq 2k-1$.
\end{theorem}

\begin{proof}
	Assume an implementation $I=[t_1, \ldots, t_{i_1}=y_1, \ldots, t_{i_k}=y_k]$ of $M$ corresponds to $T$. By Lemma \ref{le:3-1}, we have $i_1 \geq k-1$. Since each output variable $t_{i_p}=y_p$ ($2 \leq p \leq k$) must be generated, it follows that $i_p - i_{p-1} \geq 1$ ($2 \leq p \leq k$). Thus, $i_k \geq 2k-2$. Equality holds only if $i_1 = k-1$ and $i_p - i_{p-1} = 1$ ($2 \leq p \leq k$). We show that this situation does not arise, implying $i_k \geq 2k-1$.
	
	If $i_1 = k-1$ and $i_p - i_{p-1} = 1$ ($2 \leq p \leq k$), by Lemma \ref{le:3-1}, each $x_q$ appears exactly once in the expansion of $t_{i_1}=y_1$. Consider $y_2 = t_{i_2} = a_{i_2}t_{i_{2,1}} \oplus b_{i_2}t_{i_{2,2}}$ with $i_{2,1} < i_{2,2}$. Then, $t_{i_{2,2}} \prec y_1$ or $t_{i_{2,2}} = y_1$ since $y_2$ is generated immediately after $y_1$. We analyze two cases:
	
	\begin{enumerate}
		\item Consider the case where $t_{i_{2,1}}$ and $t_{i_{2,2}}$ are not comparable under the partial order relation $\prec$. It follows that $t_{i_{2,2}} \prec y_1$. Suppose that $y_1 = t_{i_1} = a_{i_1}t_{i_{2,1}} \oplus b_{i_1}t_{i_{2,2}}$. If this is not the case, then $y_1 = t_{i_1} = a_{i_1}t_{i_{1,1}} \oplus b_{i_1}t_{i_{1,2}}$, which implies that either $t_{i_{2,1}} \prec t_{i_{1,1}}$ or $t_{i_{2,2}} \prec t_{i_{1,1}}$ or $t_{i_{2,1}} \prec t_{i_{1,2}}$ or $t_{i_{2,2}} \prec t_{i_{1,2}}$. This would lead to the omission of some $x_q$ when expanding $y_2$ in terms of the input variables, contradicting Lemma \ref{le:3-1}. Assuming $y_1 = t_{i_1} = a_{i_1}t_{i_{2,1}} \oplus b_{i_1}t_{i_{2,2}}$, since $k \geq 3$, at least one of $t_{i_{2,1}}$ or $t_{i_{2,2}}$ must contain more than two input variables. Let $x_m$ and $x_n$ be two such variables. Considering the $2 \times 2$ matrix $M_2$ of $M$ corresponding to $x_m, x_n; y_1, y_2$, we find that the determinant of $M_2$ is 0, contradicting the fact that $M$ is an MDS matrix. Therefore, the first case does not hold.
		
		\item Now consider the case where $t_{i_{2,1}} \prec t_{i_{2,2}}$. In this scenario, it must be that $t_{i_{2,2}} = y_1$. Otherwise, there would be some $x_q$ missing when expanding $y_2$ in terms of the input variables, contradicting Lemma \ref{le:3-1}. Assuming $t_{i_{2,2}} = y_1$, if $t_{i_{2,1}}$ contains more than two input variables, let $x_m$ and $x_n$ be two such variables. If not, choose $x_m$ and $x_n$ from the missing terms of $t_{i_{2,1}}$. Considering the $2 \times 2$ submatrix $M_2'$ of $M$ corresponding to $x_m, x_n; y_1, y_2$, we find that the determinant of $M_2'$ is 0, contradicting the fact that $M$ is an MDS matrix. Therefore, the second case does not hold.
	\end{enumerate}
	
	In both cases, we reach a contradiction. Therefore, $i_k \geq 2k-1$.
\end{proof}

\subsection{An algorithm for searching the simplest tree}

In the preceding section, we introduced the concept of the implementation tree. We can leverage this structure to construct the $k \times k$ MDS matrix over $R_n$. As exemplified in Example \ref{exg:2}, the cost of the word-xor operation constitutes a significant portion of the overall cost in implementing an MDS matrix. Therefore, our approach prioritizes searching for implementation trees with minimal capacity. Subsequently, we assign parameters for scalar multiplications to derive the specific implementation, and ultimately yield a low-cost MDS matrix.

Our algorithm is divided into two main parts:
\begin{enumerate}
	\item Firstly we treat the parameters $a, b$ of each node $Y = a \cdot Y_1 \oplus b \cdot Y_2$ in an implementation tree as undetermined coefficients. We perform a traversal search to identify all the simplest trees which have the potential to form MDS matrices.
	
	\item Secondly we assign values to the undetermined parameters within the simplest trees. This instantiation process will result in a specific MDS matrix.
\end{enumerate}

By Theorem \ref{th:3-2}, for a given $k \times k$ ($k \geq 3$) MDS matrix $M$ over $R_n$, any implementation $T = [T_1, \cdots, y_1 = T_{i_1}, \cdots, y_k = T_{i_k}]$ of $M$ satisfies the conditions $i_k \geq 2k - 1$, $i_1 \geq k - 1$, and $i_p - i_{p-1} \geq 1$ ($2 \leq p \leq k$). Therefore, we commence our search with $i_k = 2k - 1$ and look for the simplest trees with this capacity. If there does not exist such a simplest tree, we increment $i_k$ and repeat the search.

Next, we outline the procedure for finding a possible implementation tree with a capacity of $i_k$ as below:
\begin{enumerate}
	\item We solve the indefinite equation $s_1 + s_2 + \cdots + s_k = i_k$ under the constraints $s_1 \geq k - 1$ and $s_i \geq 1$ ($2 \leq i \leq k$) to obtain the set of all possible types, denoted by $I_{i_k}$, where each type represents the number of non-leaf nodes in its corresponding implementation tree.
	
	\item For each type $(i_1, i_2 - i_1, \cdots, i_k - i_{k-1}) \in I_{i_k}$, we initiate the search with input variables $x_1, x_2, \cdots, x_k$. Leveraging the property that each node $T_p$ in $T = [T_1, \cdots, y_1 = T_{i_1}, \cdots, y_k = T_{i_k}]$ is the sum of preceding nodes, we perform a traversal search and sequentially generate $k$ binary trees (corresponding to rows of the matrix) of the specified type. Concurrently, we verify if the matrix associated with $T$ is an MDS matrix. For simplicity, we set $x_1 = T_0, x_2 = T_{-1}, \cdots, x_k = T_{-(k-1)}$.
\end{enumerate}

To generate the first binary tree, rooted at $y_1$:
\begin{itemize}
	\item We recursively generate the tree starting from $T_1$. Each $T_1$ is expressed as $T_1 = a_1T_{p_1} \oplus b_1T_{p_2}$ with $-(k-1) \leq p_1 < p_2 \leq 0$. This process yields $C_k^2$ distinct $T_1$ options. Each $T_1$ can then serve as a child node in subsequent trees.
	
	\item We proceed to generate $T_2$ using a similar approach: $T_2 = a_2T_{p_1} \oplus b_2T_{p_2}$ with $-(k-1) \leq p_1 < p_2 \leq 1$. For each $T_1$, we obtain $C_{k+1}^2$ different $T_2$ options, resulting in a total of $C_k^2 \cdot C_{k+1}^2$ unique $(T_1, T_2)$ sequences.
	
	\item We repeat this process until we generate $y_1 = T_{i_1}$. At this point, we have generated $C_k^2 \cdot C_{k+1}^2 \cdots C_{k+i_1-1}^2$ sequences of the form $(T_1, T_2, \cdots, y_1 = T_{i_1})$.
\end{itemize}

At this stage, each generated sequence ensures that each term $T_p$ is expressed as the sum of its predecessors. However, we must further verify the following conditions to determine if these sequences can constitute the first tree in the implementation tree:

\begin{enumerate}
	\item $T_1 \prec y_1, T_2 \prec y_1, \cdots, T_{i_1-1} \prec y_1$. This condition ensures the normality of the implementation and can be checked relatively easily.
	\item Whether the $1 \times k$ matrix corresponding to $y_1$ satisfies the MDS condition (i.e., all submatrix determinants are non-zero). Since the scalar multiplications in the implementation tree lack specific values, the elements of the $1 \times k$ matrix are polynomials in $a_1, b_1, \cdots, a_{i_1}, b_{i_1}$. In practice, we verify if the determinants of all submatrices of the $1 \times k$ matrix are not divisible by 2 by means of symbolic calculation. If a determinant is divisible by 2, it implies that regardless of the specific values assigned to the parameters $a_1, b_1, \cdots, a_{i_1}, b_{i_1}$, they cannot form an MDS matrix. Such sequences are subsequently discarded.
\end{enumerate}

Let $S_1$ denote the set of all sequences that pass the above test. Each sequence $(T_{-(k-1)},\cdots,T_{i_1})$ in $S_1$ serves as the foundation for generating the second tree in the implementation tree hierarchy. 
Following a similar approach to the generation of the first tree, we obtain $\vert S_1\vert\cdot C_{k+i_1}^2\cdots C_{k+i_1+i_2-1}^2$ candidate sequences. Subsequently, we evaluate whether they meet two conditions:
\begin{enumerate}
	\item $T_{i_1+1} \prec y_2, T_{i_1+2} \prec y_2, \cdots, T_{i_1+i_2-1} \prec y_2$
	\item Whether the $2\times k$ matrix associated with $y_1, y_2$ satisfies the MDS condition, i.e., all submatrix determinants are nonzero.
\end{enumerate}

The set of sequences that fulfill the above two conditions is denoted as $S_2$. By extension, all subsequent trees in the implementation tree hierarchy can be generated analogously. This iterative process terminates in two scenarios:
\begin{enumerate}
	\item At row $p$, if $S_p = \emptyset$, it indicates the absence of an implementation tree corresponding to the type $(i_1, i_2-i_1, \cdots, i_k-i_{k-1})$ (or, more precisely, the absence of an implementation tree for types whose first $p$ terms equal $(i_1, i_2-i_1, \cdots, i_p-i_{p-1})$).
	\item If $S_k \neq \emptyset$, then we obtain a successful construction of an implementation tree for the MDS matrix of the given type. Since the search is exhaustive, we can obtain all implementation trees corresponding to the type $(i_1, i_2-i_1, \cdots, i_k-i_{k-1})$.
\end{enumerate}

Finally, by iterating over all possible values of $(i_1, i_2-i_1, \cdots, i_k-i_{k-1})$ in $I_{i_k}$, we either obtain all implementation trees with capacity $i_k$ or determine that no such trees exist. At this point, we increment $i_k$ by 1 and repeat all the above procedures.

\begin{algorithm}[H]
	\caption{Search for the simplest trees that can form an MDS matrix}
	\label{alg:search-simplest-trees}
	\LinesNumbered
	\KwIn{The size $k$ of the matrix $M$ on $R_n$}
	\KwOut{The minimum capacity $i_k$ of all implementations for a $k\times k$ MDS matrix on $R_n$. The simplest trees that can form MDS matrices and their corresponding types $(a_1, a_2, \ldots, a_k)$ satisfy $a_1 \geq k-1$, $a_i \geq 1$ for $2 \leq i \leq k$, and $\sum_{i=1}^{k} a_i = i_k$.}
	
	Initialize $i_k \gets 2k-1$\;
	
	Solve the indefinite equation $\sum_{i=1}^{k} s_i = i_k$ under the constraints $s_1 \geq k-1$ and $s_i \geq 1$ for $i \geq 2$ to obtain the set $I_{i_k}$ of all possible types\;
	
	\ForEach{$(a_1, a_2, \ldots, a_k) \in I_{i_k}$}{
		Sequentially generate each binary tree in the implementation tree\;
		\ForEach{generated binary tree (i.e., a row in the matrix)}{
			Check if the matrix formed by this row and the previous rows meets the MDS conditions\;
			\If{the MDS conditions are not met}{
				Break and continue with the next type in $I_{i_k}$\;
			}
		}
		\If{There exists a simplest tree belonging to type $(a_1, a_2, \ldots, a_k)$}{
			Record the current type and all its corresponding simplest tree\;
		}
	}
	
	\If{no corresponding simplest tree is found for any $(a_1, a_2, \ldots, a_k) \in I_{i_k}$}{
		Increment $i_k \gets i_k + 1$\;
		Go back to Step 2\;
	}
	
	\Return{$i_k$, all recorded simplest trees, and their corresponding types $(a_1, a_2, \ldots, a_k)$.}
\end{algorithm}

\subsection{Feasible structures}

We executed Algorithm \ref{alg:search-simplest-trees} on matrices of sizes $3\times 3$, $4\times 4$, and $5\times 5$ defined over the ring $R_n$ and obtained the following results.

\begin{table}[htbp]
	\centering
	\caption{Outcomes of Algorithm 1}
	\label{tab:algorithm-outcomes}
	\begin{tabular}{ccp{5cm}}
		\hline
		Matrix Size & Minimum XOR Count & Feasible Types \\
		\hline
		$2\times 2$ & 2 & (1,1) \\
		$3\times 3$ & 5 & (2,2,1), (3,1,1) \\
		$4\times 4$ & 8 & (3,3,1,1), (4,2,1,1) \\
		$5\times 5$ & 12 & (4,5,1,1,1), (4,4,2,1,1), (4,4,1,2,1), (5,4,1,1,1), (5,3,2,1,1), (5,3,1,2,1), (6,3,1,1,1), (6,2,2,1,1), (6,2,1,2,1) \\
		\hline
	\end{tabular}
\end{table}

Of particular interest is the $4\times 4$ MDS matrix over $R_n$. Guided by Theorem \ref{th:3-2}, our traversal search began with an implementation tree of capacity 7. Utilizing Algorithm \ref{alg:search-simplest-trees}, we swiftly discovered, within milliseconds, that no tree of this capacity existed. Subsequently, we tested the trees with a capacity of 8 and successfully located some viable types. A comprehensive list of all the simplest trees with a capacity of 8 is provided in Appendix A.

\subsection{The $4\times 4$ MDS matrices with the lowest cost on $R_8$}

In this section we establish the ring $R_8$ as $F_2[\alpha]$, where $\alpha$ denotes the companion matrix associated with the polynomial $\alpha^8 + \alpha^2 + 1$. Explicitly, $\alpha$ is represented by the following matrix:
\begin{equation*}
	\alpha =
	\begin{bmatrix}
		0 & 0 & 0 & 0 & 0 & 0 & 0 & 1 \\
		1 & 0 & 0 & 0 & 0 & 0 & 0 & 0 \\
		0 & 1 & 0 & 0 & 0 & 0 & 0 & 1 \\
		0 & 0 & 1 & 0 & 0 & 0 & 0 & 0 \\
		0 & 0 & 0 & 1 & 0 & 0 & 0 & 0 \\
		0 & 0 & 0 & 0 & 1 & 0 & 0 & 0 \\
		0 & 0 & 0 & 0 & 0 & 1 & 0 & 0 \\
		0 & 0 & 0 & 0 & 0 & 0 & 1 & 0 \\
	\end{bmatrix}.
\end{equation*}

The choice of this ring is motivated by the fact that it contains a unique zero divisor, $\alpha^4 + \alpha + 1$, and the multiplication by $\alpha$ has a minimal cost of just one xor operation.

To construct MDS matrices with the lowest possible cost, we must assign values to the simplest trees described in section 3.3. During the assignment process, we restrict the values of all parameters to the set $\{1, \alpha, \alpha^2, \alpha^3, \alpha^{-1}, \alpha^{-2}, \alpha^{-3}\}$, whose elements have the cost of no more than 2 xor operations. We go through all cases with low costs and obtain 60 distinct $4 \times 4$ MDS matrices over $F_2[\alpha]$ with 67 xor operations. 30 of them are listed in the following tables, and another 30 MDS matrices can be got by replacing $\alpha$ with $\alpha^{-1}$.

\begin{table*}[!ht]
	\caption{MDS matrices and implementations in the simplest tree 1}
	\label{tab:Type1}
	\scriptsize
	\centering
	\resizebox{\textwidth}{26mm}{
		\begin{tabular}{c|c|c}
			\hline
			Number & Matrix & Implementation \\
			\hline
			\multirow{4}{*}{1}&\multirow{4}{*}{\scalebox{0.8}{$
				\begin{array}{cccc}
					\alpha & \alpha & \alpha & 1\\
					\alpha+\alpha^2 & \alpha^2 & 1+\alpha^2 & 1\\
					\alpha^{-1}+\alpha+\alpha^2 & \alpha^{-1}+\alpha^2 & \alpha^2 & 1\\
					\alpha^{-1}+1+\alpha^2 & \alpha^{-1}+\alpha+\alpha^2 & \alpha+\alpha^2 & 1
				\end{array}
				$}}
			&$ x_5 =  x_1 \oplus  x_2,\ x_6 =  x_3 \oplus  x_5,\ x_7 =  x_4 \oplus \alpha\cdot x_6 = y_1 $\\
			& & $ x_{8} =  x_1 \oplus  \alpha\cdot x_6,\ x_{9} =  x_4 \oplus \alpha\cdot x_8,\ x_{10} =  x_3 \oplus  x_{9} = y_2 $\\
			& & $ x_{11} = \alpha^{-1}\cdot x_5 \oplus  x_{9} = y_3 $\\
			& & $ x_{12} =  x_8 \oplus  x_{11} = y_4 $\\
			\hline
			\multirow{4}{*}{2}&\multirow{4}{*}{\scalebox{0.8}{$\begin{array}{cccc} 1 & 1 & 1 & 1\\
					\alpha+\alpha^2 & \alpha^2 & 1+\alpha^2 & \alpha\\
					\alpha^{-1}+\alpha+\alpha^2 & \alpha^{-1}+\alpha^2 & \alpha^2 & \alpha\\
					\alpha^{-1}+1+\alpha^2 & \alpha^{-1}+\alpha+\alpha^2 & \alpha+\alpha^2 & \alpha
				\end{array}
				$}}
			& $ x_5 =  x_1 \oplus  x_2,\ x_6 =  x_3 \oplus  x_5,\ x_7 =  x_4 \oplus  x_6 = y_1 $\\
			
			& & $x_{8} =  x_1 \oplus  \alpha\cdot x_6,\ x_{9} =  x_4 \oplus  x_8,\ x_{10} =  x_3 \oplus \alpha\cdot x_{9} = y_2 $\\
			& & $x_{11} = \alpha^{-1}\cdot x_5 \oplus \alpha\cdot x_{9} = y_3$\\
			& & $x_{12} =  x_8 \oplus  x_{11} = y_4$\\
			\hline
			\multirow{4}{*}{3}&\multirow{4}{*}{\scalebox{0.8}{$ \begin{array}{cccc}
					\alpha & \alpha & 1 & 1\\
					\alpha^2+\alpha^3 & \alpha^3 & 1+\alpha^2 & \alpha\\
					1+\alpha^2+\alpha^3 & 1+\alpha^3 & \alpha^2 & \alpha\\
					1+\alpha+\alpha^3 & 1+\alpha^2+\alpha^3 & \alpha+\alpha^2 & \alpha
				\end{array}
				$}}
			& $ x_5 =  x_1 \oplus  x_2,\ x_6 =  x_3 \oplus \alpha\cdot x_5,\ x_7 =  x_4 \oplus  x_6 = y_1 $\\
			
			& & $x_{8} =  x_1 \oplus  x_6,\ x_{9} =  x_4 \oplus \alpha\cdot x_8,\ x_{10} =  x_3 \oplus \alpha\cdot x_{9} = y_2 $\\
			& & $x_{11} =  x_5 \oplus \alpha\cdot x_{9} = y_3 $\\
			& & $x_{12} = \alpha\cdot x_8 \oplus  x_{11} = y_4$\\
			
			\hline
			\multirow{4}{*}{4}&\multirow{4}{*}{\scalebox{0.8}{$\begin{array}{cccc}
					\alpha & \alpha & \alpha & 1\\
					\alpha+\alpha^2 & \alpha^2 & 1+\alpha^2 & 1\\
					1+\alpha^2+\alpha^3 & 1+\alpha^3 & \alpha^3 & \alpha\\
					1+\alpha+\alpha^3 & 1+\alpha^2+\alpha^3 & \alpha^2+\alpha^3 & \alpha
				\end{array}
				$}}
			& $x_5 =  x_1 \oplus  x_2,\ x_6 =  x_3 \oplus  x_5,\ x_7 =  x_4 \oplus \alpha\cdot x_6 = y_1$\\
			
			& & $x_{8} =  x_1 \oplus \alpha\cdot x_6,\ x_{9} =  x_4 \oplus \alpha\cdot x_8,\ x_{10} =  x_3 \oplus  x_{9} = y_2$\\
			& & $x_{11} =  x_5 \oplus \alpha\cdot x_{9} = y_3$\\
			& & $x_{12} = \alpha\cdot x_8 \oplus  x_{11} = y_4$\\
			\hline
	\end{tabular}}
\end{table*}

\begin{table*}[!ht]
	\caption{MDS matrices and implementations in simplest the tree 2}
	\label{tab:Type2}
	\scriptsize
	\centering
	\resizebox{\textwidth}{26mm}{
		\begin{tabular}{c|c|c}
			\hline
			Number & Matrix & Implementation \\
			\hline
			\multirow{4}{*}{5}&\multirow{4}{*}{\scalebox{0.8}{$
				\begin{array}{cccc}
					1 & 1 & 1 & 1\\
					\alpha^2 & \alpha+\alpha^2 & 1+\alpha & \alpha\\
					\alpha^{-1}+\alpha^2 & \alpha^{-1}+\alpha+\alpha^2 & \alpha^{-1}+\alpha & \alpha\\
					\alpha^{-1}+1+\alpha+\alpha^2 & \alpha^{-1}+1+\alpha^2 & \alpha^{-1}+1+\alpha & 1+\alpha
				\end{array}
				$}}
			&$x_5 =  x_1 \oplus  x_2,\ x_6 =  x_3 \oplus  x_5,\ x_7 =  x_4 \oplus  x_6 = y_1$\\
			& & $x_{8} =  \alpha\cdot x_5 \oplus  x_7,\ x_{9} =  x_1 \oplus  x_8,\ x_{10} =  x_3 \oplus \alpha\cdot x_{9} = y_2$\\
			& & $x_{11} = \alpha^{-1}\cdot x_6 \oplus \alpha\cdot x_{9} = y_3$\\
			& & $x_{12} =  x_8 \oplus  x_{11} = y_4$\\
			\hline
			\multirow{4}{*}{6}&\multirow{4}{*}{\scalebox{0.8}{$ \begin{array}{cccc}
					\alpha & \alpha & 1 & 1\\
					\alpha^{-1} & \alpha^{-1}+1 & \alpha^{-1}+1 & \alpha^{-1}\\
					\alpha^{-1}+\alpha^2 & \alpha^{-1}+1+\alpha^2 & \alpha^{-1}+\alpha & \alpha^{-1}\\
					\alpha^{-1}+1+\alpha+\alpha^2 & \alpha^{-1}+\alpha+\alpha^2 & \alpha^{-1}+1+\alpha & 1+\alpha^{-1}
				\end{array}
				$}}
			& $x_5 =  x_1 \oplus  x_2,\ x_6 = x_3 \oplus \alpha\cdot x_5,\ x_7 =  x_4 \oplus  x_6 = y_1$\\
			
			& & $x_{8} =   x_5 \oplus  x_7,\ x_{9} =  x_1 \oplus \alpha^{-1}\cdot x_8,\ x_{10} =  x_3 \oplus  x_{9} = y_2$\\
			& & $x_{11} = \alpha\cdot x_6 \oplus  x_{9} = y_3$\\
			& & $x_{12} =  x_8 \oplus  x_{11} = y_4$\\
			\hline
			\multirow{4}{*}{7}&\multirow{4}{*}{\scalebox{0.8}{$ \begin{array}{cccc}
					\alpha^{-1} & \alpha^{-1} & 1 & 1\\
					\alpha & 1+\alpha & 1+\alpha & \alpha\\
					\alpha^{-1}+\alpha^2 & \alpha^{-1}+\alpha+\alpha^2 & 1+\alpha^2 & \alpha^2\\
					\alpha^{-1}+\alpha^3 & \alpha^{-1}+\alpha^2+\alpha^3 & 1+\alpha+\alpha^3 & 1+\alpha^3
				\end{array}
				$}}
			& $x_5 =  x_1 \oplus  x_2,\ x_6 =  x_3 \oplus \alpha^{-1}\cdot x_5,\ x_7 =  x_4 \oplus  x_6 = y_1$\\
			& & $x_{8} =   x_5 \oplus  x_7,\ x_{9} =  x_1 \oplus \alpha\cdot x_8,\ x_{10} =  x_3 \oplus  x_{9} = y_2$\\
			& & $x_{11} =  x_6 \oplus  \alpha\cdot x_{9} = y_3$\\
			& & $x_{12} =  x_8 \oplus  \alpha\cdot x_{11} = y_4$\\
			
			\hline
			\multirow{4}{*}{8}&\multirow{4}{*}{\scalebox{0.8}{$ \begin{array}{cccc}
					\alpha^{-1} & \alpha^{-1} & \alpha^{-1} & 1\\
					\alpha^2 & \alpha+\alpha^2 & 1+\alpha & \alpha^2\\
					\alpha^{-1}+\alpha^2 & \alpha^{-1}+\alpha+\alpha^2 & \alpha^{-1}+\alpha & \alpha^2\\
					\alpha^{-1}+\alpha^3 & \alpha^{-1}+\alpha^2+\alpha^3 & \alpha^{-1}+1+\alpha^2 & 1+\alpha^3
				\end{array}
				$}}
			& $x_5 =  x_1 \oplus  x_2,\ x_6 =  x_3 \oplus  x_5,\ x_7 =  x_4 \oplus \alpha^{-1}\cdot x_6 = y_1$\\
			
			& & $x_{8} =   x_5 \oplus  x_7,\ x_{9} =  x_1 \oplus \alpha\cdot x_8,\ x_{10} =  x_3 \oplus \alpha\cdot x_{9} = y_2$\\
			& & $x_{11} = \alpha^{-1}\cdot x_6 \oplus  \alpha\cdot x_{9} = y_3$\\
			& & $x_{12} = \alpha\cdot x_8 \oplus   x_{11} = y_4$\\
			\hline
	\end{tabular}}
\end{table*}

\begin{table*}[!ht]
	\caption{MDS matrices and implementations in the simplest tree 3}
	\label{tab:Type3}
	\scriptsize
	\centering
	\resizebox{\textwidth}{26mm}{
		\begin{tabular}{c|c|c}
			\hline
			Number & Matrix & Implementation \\
			\hline
			\multirow{4}{*}{9}&\multirow{4}{*}{\scalebox{0.8}{$
					\begin{array}{cccc}
						\alpha & \alpha & \alpha & 1\\
						\alpha^2 & \alpha+\alpha^2 & 1+\alpha+\alpha^2 & 1\\
						\alpha^{-1}+\alpha^2 & \alpha^{-1}+\alpha+\alpha^2 & \alpha+\alpha^2 & 1\\
						\alpha^{-1}+1+\alpha+\alpha^2 & \alpha^{-1}+\alpha^2 & \alpha^2 & 1
					\end{array}
					$}}
			&$x_5 =  x_1 \oplus  x_2,\ x_6 =  x_3 \oplus  x_5,\ x_7 =  x_4 \oplus \alpha\cdot x_6 = y_1$\\
			& & $x_{8} =   x_1 \oplus \alpha\cdot x_6,\ x_{9} =  x_7 \oplus \alpha\cdot x_8,\ x_{10} =  x_3 \oplus  x_{9} = y_2$\\
			& & $x_{11} = \alpha^{-1}\cdot x_5 \oplus  x_{9} = y_3$\\
			& & $x_{12} =  x_8 \oplus  x_{11} = y_4$\\
			\hline
			\multirow{4}{*}{10}&\multirow{4}{*}{\scalebox{0.8}{$\begin{array}{cccc} 1 & 1 & 1 & 1\\
						\alpha^2 & \alpha+\alpha^2 & 1+\alpha+\alpha^2 & \alpha\\
						\alpha^{-1}+\alpha^2 & \alpha^{-1}+\alpha+\alpha^2 & \alpha+\alpha^2 & \alpha\\
						\alpha^{-1}+1+\alpha+\alpha^2 & \alpha^{-1}+\alpha^2 & \alpha^2 & \alpha
					\end{array}
					$}}
			& $x_5 =  x_1 \oplus  x_2,\ x_6 =  x_3 \oplus  x_5,\ x_7 =  x_4 \oplus  x_6 = y_1$\\
			
			& & $x_{8} =   x_1 \oplus \alpha\cdot x_6,\ x_{9} =  x_7 \oplus  x_8,\ x_{10} =  x_3 \oplus \alpha\cdot x_{9} = y_2$\\
			& & $x_{11} = \alpha^{-1}\cdot x_5 \oplus \alpha\cdot x_{9} = y_3$\\
			& & $x_{12} =  x_8 \oplus  x_{11} = y_4$\\
			\hline
			\multirow{4}{*}{11}&\multirow{4}{*}{\scalebox{0.8}{$\begin{array}{cccc}
						\alpha & \alpha & 1 & 1\\
						\alpha^3 & \alpha^2+\alpha^3 & 1+\alpha+\alpha^2 & \alpha\\
						1+\alpha^3 & 1+\alpha^2+\alpha^3 & \alpha+\alpha^2 & \alpha\\
						1+\alpha+\alpha^2+\alpha^3 & 1+\alpha^3 & \alpha^2 & \alpha
					\end{array}
					$}}
			& $x_5 =  x_1 \oplus  x_2,\ x_6 =  x_3 \oplus \alpha\cdot x_5,\ x_7 =  x_4 \oplus  x_6 = y_1$\\
			& & $x_{8} =   x_1 \oplus  x_6,\ x_{9} =  x_7 \oplus \alpha\cdot x_8,\ x_{10} =  x_3 \oplus \alpha\cdot x_{9} = y_2$\\
			& & $x_{11} =  x_5 \oplus \alpha\cdot x_{9} = y_3$\\
			& & $x_{12} =  \alpha\cdot x_8 \oplus  x_{11} = y_4$\\
			
			\hline
			\multirow{4}{*}{12}&\multirow{4}{*}{\scalebox{0.8}{$\begin{array}{cccc}
						\alpha & \alpha & \alpha & 1\\
						\alpha^2 & \alpha+\alpha^2 & 1+\alpha+\alpha^2 & 1\\
						1+\alpha^3 & 1+\alpha^2+\alpha^3 & \alpha^2+\alpha^3 & \alpha\\
						1+\alpha+\alpha^2+\alpha^3 & 1+\alpha^3 & \alpha^3 & \alpha
					\end{array}
					$}}
			& $x_5 =  x_1 \oplus  x_2,\ x_6 =  x_3 \oplus  x_5,\ x_7 =  x_4 \oplus \alpha\cdot x_6 = y_1$\\
			
			& & $x_{8} =   x_1 \oplus \alpha\cdot x_6,\ x_{9} =  x_7 \oplus \alpha\cdot x_8,\ x_{10} =  x_3 \oplus x_{9} = y_2$\\
			& & $x_{11} =  x_5 \oplus \alpha\cdot x_{9} = y_3$\\
			& & $x_{12} =  \alpha\cdot x_8 \oplus  x_{11} = y_4$\\
			\hline
	\end{tabular}}
\end{table*}

\begin{table*}[!ht]
	\caption{MDS matrices and implementations in the simplest tree 4}
	\label{tab:Type4}
	\scriptsize
	\centering
	\resizebox{\textwidth}{26mm}{
		\begin{tabular}{c|c|c}
			\hline
			Number & Matrix & Implementation \\
			\hline
			\multirow{4}{*}{13}&\multirow{4}{*}{\scalebox{0.8}{$ \begin{array}{cccc}
						\alpha & \alpha & 1 & 1\\
						\alpha^{-1}+1 & \alpha^{-1} & 1 & \alpha^{-1}\\
						\alpha^{-1}+1+\alpha^2 & \alpha^{-1}+\alpha^2 & \alpha & \alpha^{-1}\\
						\alpha^{-1}+\alpha^2 & \alpha^{-1}+1+\alpha^2 & \alpha & \alpha^{-1}+1
					\end{array}
					$}}
			&$x_5 =  x_1 \oplus  x_2,\ x_6 =  x_3 \oplus \alpha\cdot x_5,\ x_7 =  x_4 \oplus  x_6 = y_1$\\
			& & $x_{8} =   x_4 \oplus  x_5,\ x_{9} =  x_1 \oplus \alpha^{-1}\cdot x_8,\  x_{10} =  x_3 \oplus  x_{9} = y_2$\\
			& & $x_{11} =  \alpha\cdot x_6 \oplus  x_{9} = y_3$\\
			& & $x_{12} =  x_8 \oplus  x_{11} = y_4$\\
			\hline
			\multirow{4}{*}{14}&\multirow{4}{*}{\scalebox{0.8}{$\begin{array}{cccc} 1 & 1 & 1 & 1\\
							\alpha+\alpha^2 & \alpha^2 & 1 & \alpha\\
							\alpha^{-1}+\alpha+\alpha^2 & \alpha^{-1}+\alpha^2 & \alpha^{-1} & \alpha\\
							\alpha^{-1}+\alpha^2 & \alpha^{-1}+\alpha+\alpha^2 & \alpha^{-1} & 1+\alpha
						\end{array}
						$}}
				& $x_5 =  x_1 \oplus  x_2,\ x_6 =  x_3 \oplus  x_5,\ x_7 =  x_4 \oplus  x_6 = y_1$\\
				& & $x_{8} =   x_4 \oplus \alpha\cdot x_5,\ x_{9} =  x_1 \oplus  x_8,\  x_{10} =  x_3 \oplus \alpha\cdot x_{9} = y_2$\\
				& & $x_{11} =  \alpha^{-1}\cdot x_6 \oplus \alpha\cdot x_{9} = y_3$\\
				& & $x_{12} =  x_8 \oplus  x_{11} = y_4$\\
				\hline
				\multirow{4}{*}{15}&\multirow{4}{*}{\scalebox{0.8}{$ \begin{array}{cccc} \alpha^{-1} & \alpha^{-1} & 1 & 1\\
								1+\alpha& \alpha & 1 & \alpha\\
								\alpha^{-1}+\alpha+\alpha^2 & \alpha^{-1}+\alpha^2 & 1 & \alpha^2\\
								\alpha^{-1}+\alpha^2 & \alpha^{-1}+\alpha+\alpha^2 & 1 & \alpha+\alpha^2
							\end{array}
							$}}
					& $x_5 =  x_1 \oplus  x_2,\ x_6 =  x_3 \oplus \alpha^{-1}\cdot x_5,\ x_7 =  x_4 \oplus  x_6 = y_1$\\
					& & $x_{8} =   x_4 \oplus  x_5,\ x_{9} =  x_1 \oplus \alpha\cdot x_8,\  x_{10} =  x_3 \oplus  x_{9} = y_2$\\
					& & $x_{11} =   x_6 \oplus \alpha\cdot x_{9} = y_3$\\
					& & $x_{12} =  \alpha\cdot x_8 \oplus  x_{11} = y_4$\\
					
					\hline
					\multirow{4}{*}{16}&\multirow{4}{*}{\scalebox{0.8}{$\begin{array}{cccc} \alpha^{-1} & \alpha^{-1} & \alpha^{-1} & 1\\
								\alpha+\alpha^2 & \alpha^2 & 1 & \alpha^2\\
								\alpha^{-1}+\alpha+\alpha^2 & \alpha^{-1}+\alpha^2 & \alpha^{-1} & \alpha^2\\
								\alpha^{-1}+\alpha^2 & \alpha^{-1}+\alpha+\alpha^2 & \alpha^{-1} & \alpha+\alpha^2
							\end{array}
							$}}
					& $x_5 =  x_1 \oplus  x_2,\ x_6 =  x_3 \oplus  x_5,\ x_7 =  x_4 \oplus \alpha^{-1}\cdot x_6 = y_1$\\
					& & $x_{8} =   x_4 \oplus  x_5,\ x_{9} =  x_1 \oplus \alpha\cdot x_8,\  x_{10} =  x_3 \oplus \alpha\cdot x_{9} = y_2$\\
					& & $x_{11} =  \alpha^{-1}\cdot x_6 \oplus \alpha\cdot x_{9} = y_3$\\
					& & $x_{12} =  \alpha\cdot x_8 \oplus  x_{11} = y_4$\\
					\hline
			\end{tabular}}
		\end{table*}
		
\begin{table*}[!ht]
	\caption{MDS matrices and implementations in the simplest tree 5}
	\label{tab:Type5}
	\scriptsize
	\centering
	\resizebox{\textwidth}{26mm}{
		\begin{tabular}{c|c|c}
			\hline
			Number & Matrix & Implementation \\
			\hline
			\multirow{4}{*}{17}&\multirow{4}{*}{\scalebox{0.8}{$ \begin{array}{cccc} 1 & 1 & \alpha & \alpha\\
						\alpha^{-1}+\alpha & \alpha & \alpha^{-1}+1 & \alpha^{-1}\\
						\alpha^{-1}+1 & 1 & \alpha^{-1}+1+\alpha & \alpha^{-1}+\alpha\\
						\alpha^{-1}+1+\alpha & \alpha & \alpha^{-1} & \alpha^{-1}+1
					\end{array}
					$}}
			&$x_5 =  x_1 \oplus  x_2,\ x_6 =  x_3 \oplus  x_4,\ x_7 =  x_5 \oplus \alpha\cdot x_6 = y_1$\\
			& & $x_{8} =   x_1 \oplus  x_6,\ x_{9} =  x_3 \oplus \alpha^{-1}\cdot x_8,\  x_{10} = \alpha\cdot x_5 \oplus  x_{9} = y_2$\\
			& & $x_{11} =  x_7 \oplus  x_{9} = y_3$\\
			& & $x_{12} =  x_8 \oplus  x_{10} = y_4$\\
			\hline
			\multirow{4}{*}{18}&\multirow{4}{*}{\scalebox{0.8}{$ \begin{array}{cccc} 1 & 1 & 1 & 1\\
							\alpha^{-1}+\alpha & \alpha^{-1} & \alpha+\alpha^2 & \alpha^2\\
							1+\alpha & 1 & 1+\alpha+\alpha^2 & 1+\alpha^2\\
							\alpha^{-1}+1+\alpha & \alpha^{-1} & \alpha^2 & \alpha+\alpha^2
						\end{array}
						$}}
				& $x_5 =  x_1 \oplus  x_2,\ x_6 =  x_3 \oplus  x_4,\ x_7 =  x_5 \oplus  x_6 = y_1$\\
				
				& & $x_{8} =   x_1 \oplus \alpha\cdot x_6,\ x_{9} =  x_3 \oplus  x_8,\  x_{10} = \alpha^{-1}\cdot x_5 \oplus \alpha\cdot x_{9} = y_2$\\
				& & $x_{11} =  x_7 \oplus \alpha\cdot x_{9} = y_3$\\
				& & $x_{12} =  x_8 \oplus  x_{10} = y_4$\\
				\hline
				\multirow{4}{*}{19}&\multirow{4}{*}{\scalebox{0.8}{$\begin{array}{cccc}
							\alpha & \alpha & 1 & 1\\
							1+\alpha^2 & 1 & \alpha+\alpha^2 & \alpha^2\\
							\alpha+\alpha^2 & \alpha & 1+\alpha+\alpha^2 & 1+\alpha^2\\
							1+\alpha+\alpha^2 & 1 & \alpha^2 & \alpha+\alpha^2
						\end{array}
						$}}
				& $x_5 =  x_1 \oplus  x_2,\ x_6 =  x_3 \oplus  x_4,\ x_7 = \alpha\cdot x_5 \oplus  x_6 = y_1$\\
				& & $x_{8} =   x_1 \oplus  x_6,\ x_{9} =  x_3 \oplus \alpha\cdot x_8,\  x_{10} =  x_5 \oplus \alpha\cdot x_{9} = y_2$\\
				& & $x_{11} =  x_7 \oplus \alpha\cdot x_{9} = y_3$\\
				& & $x_{12} =  \alpha\cdot x_8 \oplus  x_{10} = y_4$\\
				
				\hline
				\multirow{4}{*}{20}&\multirow{4}{*}{\scalebox{0.8}{$\begin{array}{cccc} 1 & 1 & \alpha^{-1} & \alpha^{-1}\\
							1+\alpha^2 & 1 & \alpha+\alpha^2 & \alpha^2\\
							1+\alpha & 1 & \alpha^{-1}+1+\alpha & \alpha^{-1}+\alpha\\
							1+\alpha+\alpha^2 & 1 & \alpha^2 & \alpha+\alpha^2
						\end{array}
						$}}
				& $x_5 =  x_1 \oplus  x_2,\ x_6 =  x_3 \oplus  x_4,\ x_7 =  x_5 \oplus \alpha^{-1}\cdot x_6 = y_1$\\
				& & $x_{8} =   x_1 \oplus  x_6,\ x_{9} =  x_3 \oplus \alpha\cdot x_8,\  x_{10} =  x_5 \oplus \alpha\cdot x_{9} = y_2$\\
				& & $x_{11} =  x_7 \oplus  x_{9} = y_3$\\
				& & $x_{12} =  \alpha\cdot x_8 \oplus  x_{10} = y_4$\\
				\hline
		\end{tabular}}
	\end{table*}
	
	\begin{table*}[!ht]
	\caption{MDS matrices and implementations in the simplest tree 6}
	\label{tab:Type7}
	\scriptsize
	\centering
	\resizebox{\textwidth}{13.5mm}{
		\begin{tabular}{c|c|c}
			\hline
			Number & Matrix & Implementation \\
			\hline
			\multirow{4}{*}{21}&\multirow{4}{*}{\scalebox{0.8}{$\begin{array}{cccc} 1+\alpha & 1 & \alpha & \alpha\\
							\alpha & \alpha & \alpha^{-1} & \alpha^{-1}+1\\
							1 & 1+\alpha & \alpha & 1+\alpha\\
							1+\alpha & \alpha & \alpha^{-1}+1 & \alpha^{-1}
						\end{array}
						$}}
				&$x_5 =  x_1 \oplus  x_2,\ x_6 =  x_3 \oplus  x_4,\ x_7 =  x_1 \oplus  x_6,\  x_8 =  x_5 \oplus \alpha\cdot x_7 = y_1$\\
				& & $x_{9} =    x_4 \oplus \alpha^{-1}\cdot x_5,\  x_{10} = \alpha\cdot x_6 \oplus  x_{9} = y_2$\\
				& & $x_{11} =   x_8 \oplus  x_{9} = y_3$\\
				& & $x_{12} =   x_7 \oplus  x_{10} = y_4$\\
				\hline
				\multirow{4}{*}{22}&\multirow{4}{*}{\scalebox{0.8}{$\begin{array}{cccc} 1+\alpha & 1 & \alpha & \alpha\\
							\alpha^2 & \alpha^2 & 1 & 1+\alpha\\
							1 & 1+\alpha & \alpha & 1+\alpha\\
							\alpha+\alpha^2 & \alpha^2 & 1+\alpha & 1
						\end{array}
						$}}
				& $x_5 =  x_1 \oplus  x_2,\ x_6 =  x_3 \oplus  x_4,\ x_7 =  x_1 \oplus  x_6,\  x_8 =  x_5 \oplus \alpha\cdot x_7 = y_1$\\
				& & $x_{9} =    x_4 \oplus \alpha\cdot x_5,\  x_{10} =  x_6 \oplus \alpha\cdot x_{9} = y_2$\\
				& & $x_{11} =   x_8 \oplus  x_{9} = y_3$\\
				& & $x_{12} =   \alpha\cdot x_7 \oplus  x_{10} = y_4$\\
				\hline
		\end{tabular}}
	\end{table*}
	
	\begin{table*}[!ht]
		\caption{MDS matrices and implementations in the simplest tree 7}
		\label{tab:Type8}
		\scriptsize
		\centering
		\resizebox{\textwidth}{26mm}{
			\begin{tabular}{c|c|c}
				\hline
				Number & Matrix & Implementation \\
				\hline
				\multirow{4}{*}{23}&\multirow{4}{*}{\scalebox{0.8}{$\begin{array}{cccc} 1+\alpha & \alpha & \alpha & 1\\
							\alpha^{-1}+1+\alpha & \alpha^{-1}+\alpha & 1+\alpha & 1\\
							\alpha^{-1}+1+\alpha^2 & \alpha^{-1}+\alpha+\alpha^2 & \alpha+\alpha^2 & 1\\
							\alpha^{-1}+1+\alpha+\alpha^2 & \alpha^{-1}+\alpha^2 & \alpha^2 & 1
						\end{array}
						$}}
				&$x_5 =  x_1 \oplus  x_2,\ x_6 =  x_3 \oplus  x_5,\ x_7 =  x_1 \oplus \alpha\cdot x_6,\  x_8 =  x_4 \oplus  x_7 = y_1$\\
				& & $x_{9} =   \alpha^{-1}\cdot x_5 \oplus  x_8,\  x_{10} =  x_3 \oplus  x_{9} = y_2$\\
				& & $x_{11} =  \alpha\cdot x_7 \oplus  x_{9} = y_3$\\
				& & $x_{12} =  \alpha\cdot x_6 \oplus  x_{11} = y_4$\\
				\hline
				\multirow{4}{*}{24}&\multirow{4}{*}{\scalebox{0.8}{$ \begin{array}{cccc} \alpha^{-1}+1 & \alpha^{-1} & \alpha^{-1} & 1\\
							\alpha^{-1}+1+\alpha & \alpha^{-1}+\alpha & \alpha^{-1}+1 & 1\\
							\alpha^{-1}+\alpha+\alpha^2 & \alpha^{-1}+1+\alpha^2 & \alpha^{-1}+1 & \alpha\\
							\alpha^{-1}+1+\alpha+\alpha^2 & \alpha^{-1}+\alpha^2 & \alpha^{-1} & \alpha
						\end{array}
						$}}
				& $x_5 =  x_1 \oplus  x_2,\ x_6 =  x_3 \oplus  x_5,\ x_7 =  x_1 \oplus \alpha^{-1}\cdot x_6,\  x_8 =  x_4 \oplus  x_7 = y_1$\\
				& & $x_{9} =   \alpha\cdot x_5 \oplus  x_8,\  x_{10} =  x_3 \oplus  x_{9} = y_2$\\
				& & $x_{11} =   x_7 \oplus \alpha\cdot x_{9} = y_3$\\
				& & $x_{12} =  x_6 \oplus  x_{11} = y_4$\\
				\hline
				\multirow{4}{*}{25}&\multirow{4}{*}{\scalebox{0.8}{$\begin{array}{cccc} \alpha+\alpha^2 & \alpha^2 & \alpha & 1\\
							1+\alpha+\alpha^2 & 1+\alpha^2 & 1+\alpha & 1\\
							\alpha^{-1}+1+\alpha^2 & \alpha^{-1}+\alpha+\alpha^2 & 1+\alpha & \alpha^{-1}\\
							\alpha^{-1}+1+\alpha+\alpha^2 & \alpha^{-1}+\alpha^2 & \alpha & \alpha^{-1}
						\end{array}
						$}}
				& $x_5 =  x_1 \oplus  x_2,\ x_6 =  x_3 \oplus \alpha\cdot x_5,\ x_7 =  x_1 \oplus  x_6,\  x_8 =  x_4 \oplus \alpha\cdot x_7 = y_1$\\
				& & $x_{9} =    x_5 \oplus  x_8,\  x_{10} =  x_3 \oplus  x_{9} = y_2$\\
				& & $x_{11} =  \alpha\cdot x_7 \oplus \alpha^{-1}\cdot x_{9} = y_3$\\
				& & $x_{12} =  x_6 \oplus  x_{11} = y_4$\\
				
				\hline
				\multirow{4}{*}{26}&\multirow{4}{*}{\scalebox{0.8}{$ \begin{array}{cccc} \alpha+\alpha^2 & \alpha^2 & \alpha^2 & 1\\
							\alpha^{-1}+1+\alpha & \alpha^{-1}+\alpha & 1+\alpha & \alpha^{-1}\\
							\alpha^{-1}+1+\alpha^2 & \alpha^{-1}+\alpha+\alpha^2 & \alpha+\alpha^2 & \alpha^{-1}\\
							\alpha^{-1}+1+\alpha+\alpha^2 & \alpha^{-1}+\alpha^2 & \alpha^2 & \alpha^{-1}
						\end{array}
						$}}
				& $x_5 =  x_1 \oplus  x_2,\ x_6 =  x_3 \oplus  x_5,\ x_7 =  x_1 \oplus \alpha\cdot x_6,\  x_8 =  x_4 \oplus \alpha\cdot x_7 = y_1$\\
				& & $x_{9} =    x_5 \oplus  x_8,\  x_{10} =  x_3 \oplus \alpha^{-1}\cdot x_{9} = y_2$\\
				& & $x_{11} =  \alpha\cdot x_7 \oplus \alpha^{-1}\cdot x_{9} = y_3$\\
				& & $x_{12} =  \alpha\cdot x_6 \oplus  x_{11} = y_4$\\
				\hline
		\end{tabular}}
	\end{table*}
	
	\begin{table*}[!ht]
		\caption{MDS matrices and implementations in the simplest tree 8}
		\label{tab:Type6}
		\scriptsize
		\centering
		\resizebox{\textwidth}{26mm}{
			\begin{tabular}{c|c|c}
				\hline
				Number & Matrix & Implementation \\
				\hline
				\multirow{4}{*}{27}&\multirow{4}{*}{\scalebox{0.8}{$ \begin{array}{cccc} \alpha & \alpha & 1+\alpha & 1\\
							1+\alpha^2 & \alpha^2 & \alpha^2 & 1+\alpha\\
							\alpha^{-1}+1+\alpha^2 & \alpha^{-1}+\alpha^2 & \alpha^2 & \alpha\\
							\alpha^{-1}+\alpha^2 & \alpha^{-1}+1+\alpha^2 & 1+\alpha^2 & \alpha
						\end{array}
						$}}
				&$x_5 =  x_1 \oplus  x_2,\ x_6 =  x_3 \oplus  x_5,\ x_7 =  x_4 \oplus \alpha\cdot x_6,\  x_8 =  x_3 \oplus  x_7 = y_1$\\
				& & $x_{9} =    x_1 \oplus \alpha^{-1}\cdot x_7,\  x_{10} =  x_4 \oplus  x_{9} = y_2$\\
				& & $x_{11} =  \alpha\cdot x_5 \oplus  x_{9} = y_3$\\
				& & $x_{12} =   x_6 \oplus  x_{11} = y_4$\\
				\hline
				\multirow{4}{*}{28}&\multirow{4}{*}{\scalebox{0.8}{$ \begin{array}{cccc} \alpha^2 & \alpha^2 & 1+\alpha & 1\\
							\alpha+\alpha^3 & \alpha^3 & \alpha^2 & 1+\alpha\\
							1+\alpha+\alpha^3 & 1+\alpha^3 & \alpha^2 & \alpha\\
							1+\alpha^3 & 1+\alpha+\alpha^3 & 1+\alpha^2 & \alpha
						\end{array}
						$}}
				& $x_5 =  x_1 \oplus  x_2,\ x_6 =  x_3 \oplus \alpha\cdot x_5,\ x_7 =  x_4 \oplus \alpha\cdot x_6,\  x_8 =  x_3 \oplus  x_7 = y_1$\\
				& & $x_{9} =    x_1 \oplus  x_7,\  x_{10} =  x_4 \oplus \alpha\cdot x_{9} = y_2$\\
				& & $x_{11} =   x_5 \oplus \alpha\cdot x_{9} = y_3$\\
				& & $x_{12} =   x_6 \oplus  x_{11} = y_4$\\
				\hline
				\multirow{4}{*}{29}&\multirow{4}{*}{\scalebox{0.8}{$ \begin{array}{cccc} \alpha^2 & \alpha^2 & 1+\alpha & \alpha\\
							1+\alpha^2 & \alpha^2 & \alpha & 1+\alpha\\
							1+\alpha+\alpha^3 & 1+\alpha^3 & \alpha^2 & \alpha^2\\
							1+\alpha^3 & 1+\alpha+\alpha^3 & 1+\alpha^2 & \alpha^2
						\end{array}
						$}}
				& $x_5 =  x_1 \oplus  x_2,\ x_6 =  x_3 \oplus \alpha\cdot x_5,\ x_7 =  x_4 \oplus  x_6,\  x_8 =  x_3 \oplus \alpha\cdot x_7 = y_1$\\
				& & $x_{9} =    x_1 \oplus \alpha\cdot x_7,\  x_{10} =  x_4 \oplus  x_{9} = y_2$\\
				& & $x_{11} =   x_5 \oplus \alpha\cdot x_{9} = y_3$\\
				& & $x_{12} =   x_6 \oplus  x_{11} = y_4$\\
				
				\hline
				\multirow{4}{*}{30}&\multirow{4}{*}{\scalebox{0.8}{$ \begin{array}{cccc} \alpha & \alpha & 1+\alpha & 1\\
							1+\alpha^2 & \alpha^2 & \alpha^2 & 1+\alpha\\
							1+\alpha+\alpha^3 & 1+\alpha^3 & \alpha^3 & \alpha^2\\
							1+\alpha^3 & 1+\alpha+\alpha^3 & \alpha+\alpha^3 & \alpha^2
						\end{array}
						$}}
				& $x_5 =  x_1 \oplus  x_2,\ x_6 =  x_3 \oplus  x_5,\ x_7 =  x_4 \oplus \alpha\cdot x_6,\  x_8 =  x_3 \oplus  x_7 = y_1$\\
				& & $x_{9} =    x_1 \oplus \alpha\cdot x_7,\  x_{10} =  x_4 \oplus  x_{9} = y_2$\\
				& & $x_{11} =   x_5 \oplus \alpha\cdot x_{9} = y_3$\\
				& & $x_{12} =  \alpha\cdot x_6 \oplus  x_{11} = y_4$\\
				\hline
		\end{tabular}}
	\end{table*}

\begin{remark}
	The two optimal $4\times 4$ matrices, as cited in \cite{duval2018mds}, are encompassed within matrices labeled 21 and 22; the ten matrices documented in \cite{wang2023four} are integrated into matrices numbered 17, 19, 20, 21, and 22; the 52 matrices detailed in \cite{ShiWang2022} are represented in the matrices numbered 1-8 and 13-30. The rest eight matrices labeled 9-12 represent newly reported findings.
\end{remark}

Furthermore, if we replace $\alpha$ with the companion matrix of the polynomial $x^4 + x + 1$, we can naturally obtain 60 new MDS matrices with a word length of 4 that require only 35 XOR operations for their implementation.

It is important to note that substituting $\alpha$ with any invertible $8 \times 8$ matrix over the finite field $F_2$, whose minimal polynomial is $\alpha^8 + \alpha^2 + 1$, might yield some new results. However, as long as the substituted matrix can be implemented using a single xor operation, the resulting matrices will retain their MDS property and can still be implemented using 67 xor operations. It should be emphasized that the matrix obtained through such a substitution is not equivalent to the original matrix. Furthermore, while there exist numerous substitution methods that satisfy these conditions, we refrain from enumerating them here due to their extensive nature.

By slightly relaxing the requirements on the implementation cost, our method can yield a large number of $4\times 4$ MDS matrices that can be implemented with 36-41 XOR gates (for 4-bit input) and 68-80 XOR gates (for 8-bit input). If we limit the circuit depth, we can obtain MDS matrices with optimal costs for different depths. These matrices are included in Appendix B.

\subsection{Construction of $4\times 4$ Involutory MDS Matrices}

In this section, we endeavor to construct lightweight involutory MDS matrices by assigning specific parameters to the eight simplest trees obtained in Section 3.3. As a result, we have identified six involutory MDS matrices that can be efficiently implemented using only 68 XOR gates. To the best of our knowledge, this represents the current optimal result in terms of implementation cost.

Since row permutations and row multiplications affect the involutory properties of a matrix, we must consider these operations during our search. Specifically, when searching for an involutory matrix, we must account for the original $8\times 2 = 16$ parameters, as well as additional 4 parameters to represent row multiplications and $4! = 24$ possible row orders.

\begin{remark}
	It is important to note that the order of the matrix columns does not need to be considered. If $P$ and $Q$ are permutation matrices such that $(PAQ)^2 = I$, then right-multiplying both sides of the equation by $Q^{-1}$ and left-multiplying by $Q$ yields $QPAQPA = I$. This means that if there exists an involutory matrix in the equivalence class of $A$, it can be obtained simply by swapping the rows of $A$.
\end{remark}

We define the equivalence class of an involutory MDS matrix as follows: If $A$ is an involutory MDS matrix and $A' = P^{-1}AP$, where $P$ is a permutation matrix, then $A$ and $A'$ belong to the same equivalence class. It is evident that the implementation of $A'$ differs from that of $A$ only in the order of the input and output variables, and both matrices are either involutory or non-involutory.

Our search is restricted to the ring $F[x]/(x^8 + x^2 + 1)$, and the parameters are powers of $\alpha$, where $\alpha$ is the companion matrix of $x^8 + x^2 + 1$. We calculate the cost of $\alpha^n$ as $|n|$ and consider the sum of these costs as the total cost $t$. Scalar multiplications can be reused, potentially reducing the practical cost below $t$. During the search, we select $s$ (where $s \leq 6$) out of 20 parameters for assignment, fixing the unselected parameters at 1. We limit the total cost of the parameters to $t \leq 8$.

Under these conditions, we find that only the simplest trees 3 and 4 can produce involutory matrices. These involutory matrices satisfy $s \geq 4$ and $t \geq 5$. Focusing on the lightest involutory MDS matrices with $t = 5$, we identify 6 matrices from the simplest tree 3 and 12 matrices from the simplest tree 4. We provide these 18 matrices satisfying $t = 5$ below.

Simplest tree 3:
\begin{align*}
	\begin{cases}
		y_1=(a_5(a_3(a_0x_1\oplus a_1x_2)\oplus a_2x_3)\oplus a_4x_4)\cdot a_{16},\\
		y_2=(a_{10}x_3\oplus a_{11}(a_9(a_7(a_3(a_0x_1\oplus a_1x_2)\oplus a_2x_3)\oplus a_6x_1)\oplus a_8y_1))\cdot a_{17},\\
		y_3=(a_{12}(a_0x_1\oplus a_1x_2)\oplus a_{13}(a_9(a_7(a_3(a_0x_1\oplus a_1x_2)\oplus a_2x_3)\oplus a_6x_1)\oplus a_8y_1))\cdot a_{18},\\
		y_4=(a_{15}y_3\oplus a_{14}(a_7(a_3(a_0x_1\oplus a_1x_2)\oplus a_2x_3)\oplus a_6x_1))\cdot a_{19}.
	\end{cases}
\end{align*}

The six MDS matrices, associated with the simplest tree 3 and adhering to the condition $t=5$, are designated as follows:
\begin{enumerate}
	\item $a_{6}=\alpha,a_{7}=\alpha^2,a_{12}=\alpha^{-1},a_{13}=\alpha$.
	\item $a_{3}=\alpha,a_{9}=\alpha^2,a_{12}=\alpha^{-1},a_{14}=\alpha$.
	\item $a_{3}=\alpha,a_{6}=\alpha,a_{7}=\alpha,a_{9}=\alpha,a_{12}=\alpha^{-1}$.
	\item $a_{7}=\alpha,a_{9}=\alpha,a_{12}=\alpha^{-1},a_{13}=\alpha,a_{14}=\alpha$.
	\item $a_{0}=\alpha^{-1},a_{6}=\alpha,a_{7}=\alpha,a_{12}=\alpha^{-1},a_{15}=\alpha^{-1}$.
	\item $a_{0}=\alpha^{-1},a_{9}=\alpha,a_{12}=\alpha^{-1},a_{14}=\alpha,a_{15}=\alpha^{-1}$.
\end{enumerate}
The aforementioned six matrices exhibit a specific order among their rows, which is as follows:
\begin{equation*}
	P=\left[ \begin{array}{cccc}
		0 & 0 & 1 & 0\\
		0 & 0 & 0 & 1\\
		1 & 0 & 0 & 0\\
		0 & 1 & 0 & 0
	\end{array}
	\right ].
\end{equation*}

The resulting matrix $M$ is expected to have the order $(y_3, y_4, y_1, y_2)$. Among the six involutory MDS matrices with $t=5$ mentioned previously, the fourth and sixth assignments demonstrate scalar multiplication reuse, which are listed below. Since the terms multiplied by these scalars are identical, the multiplication cost is calculated only once. Therefore, the total cost of the entire matrix operation is determined to be $8 \times 8 + 4 = 68$.
\begin{equation*}
	4:\left[ \begin{array}{cccc}
		\alpha^3+\alpha^2+\alpha+\alpha^{-1} & \alpha^3+\alpha+\alpha^{-1} & \alpha^3+\alpha & \alpha\\
		\alpha^3+\alpha^{-1} & \alpha^3+\alpha^2+\alpha+\alpha^{-1} & \alpha^3+\alpha^2+\alpha & \alpha\\
		1 & 1 & 1 & 1\\
		\alpha^2+\alpha+1 & \alpha^2+1 & \alpha^2 & 1
	\end{array}
	\right ],
\end{equation*}

\begin{equation*}
	6:\left[ \begin{array}{cccc}
		\alpha+1+\alpha^{-1}+\alpha^{-2} & \alpha+1+\alpha^{-1} & \alpha+1 & 1\\
		\alpha+\alpha^{-1}+\alpha^{-2}+\alpha^{-3} & \alpha+1+\alpha^{-1}+\alpha^{-2} & \alpha+1+\alpha^{-1} & \alpha^{-1}\\
		\alpha^{-1} & 1 & 1 & 1\\
		\alpha+1+\alpha^{-1} & \alpha+1 & \alpha & 1
	\end{array}
	\right ].
\end{equation*}

Simplest tree 4:
\begin{align*}
	\begin{cases}
		y_1=(a_4(a_0x_1\oplus a_1x_2)\oplus a_5(a_2x_3\oplus a_3x_4))\cdot a_{16}, \\
		y_2=(a_{11}(a_9(a_6x_1\oplus a_7(a_2x_3\oplus a_3x_4))\oplus a_8x_3)\oplus a_{10}(a_0x_1\oplus a_1x_2))\cdot a_{17}, \\
		y_3=(a_{12}y_1\oplus a_{13}(a_9(a_6x_1\oplus a_7(a_2x_3\oplus a_3x_4))\oplus a_8x_3))\cdot a_{18}, \\
		y_4=(a_{15}y_2\oplus a_{14}(a_6x_1\oplus a_7(a_2x_3\oplus a_3x_4)))\cdot a_{19}.
	\end{cases}
\end{align*}

Here we present the twelve assignments of MDS matrices corresponding to the simplest tree 4 satisfying the condition ($t=5$). These assignments are enumerated as follows:
\begin{enumerate}
	\item $a_{4}=\alpha,a_{5}=\alpha^{-1},a_{8}=\alpha,a_{9}=\alpha^2$.
	\item $a_{6}=\alpha,a_{7}=\alpha^2,a_{10}=\alpha^{-1},a_{11}=\alpha$.
	\item $a_{5}=\alpha^{-1},a_{7}=\alpha,a_{11}=\alpha^2,a_{13}=\alpha$.
	\item $a_{4}=\alpha,a_{9}=\alpha^2,a_{10}=\alpha^{-1},a_{14}=\alpha$.
	\item $a_{4}=\alpha,a_{6}=\alpha,a_{7}=\alpha,a_{9}=\alpha,a_{10}=\alpha^{-1}$.
	\item $a_{5}=\alpha^{-1},a_{7}=\alpha,a_{8}=\alpha,a_{9}=\alpha,a_{11}=\alpha$.
	\item $a_{2}=\alpha^{-1},a_{5}=\alpha^{-1},a_{8}=\alpha,a_{9}=\alpha,a_{12}=\alpha^{-1}$.
	\item $a_{4}=\alpha,a_{5}=\alpha^{-1},a_{9}=\alpha,a_{11}=\alpha,a_{13}=\alpha$.
	\item $a_{2}=\alpha^{-1},a_{5}=\alpha^{-1},a_{11}=\alpha,a_{12}=\alpha^{-1},a_{13}=\alpha$.
	\item $a_{7}=\alpha,a_{9}=\alpha,a_{10}=\alpha^{-1},a_{11}=\alpha,a_{14}=\alpha$.
	\item $a_{0}=\alpha^{-1},a_{6}=\alpha,a_{7}=\alpha,a_{10}=\alpha^{-1},a_{15}=\alpha^{-1}$.
	\item $a_{0}=\alpha^{-1},a_{9}=\alpha,a_{10}=\alpha^{-1},a_{14}=\alpha,a_{15}=\alpha^{-1}$.
\end{enumerate}
The order between the rows for the above assignments is given by the permutation matrix:
\begin{equation*}
	P=\left[ \begin{array}{cccc}
		0 & 0 & 1 & 0\\
		0 & 0 & 0 & 1\\
		1 & 0 & 0 & 0\\
		0 & 1 & 0 & 0
	\end{array}
	\right ],
\end{equation*}
which implies that the resulting matrix $M$ should have the order $(y_2, y_4, y_1, y_3)$.

Among the twelve involutory MDS matrices presented above, the 8th, 9th, 10th, and 12th assignments exhibit scalar multiplication reuse and have the cost of 68:

\begin{equation*}
	\begin{split}
		\left[ \begin{array}{cccc}
			\alpha^2+1 & 1 & \alpha^2+\alpha & \alpha^2\\
			\alpha^2 & 1 & \alpha^2+\alpha+1 & \alpha^2+1\\
			\alpha & \alpha & \alpha^{-1} & \alpha^{-1}\\
			\alpha^2+\alpha & \alpha & \alpha^2+\alpha+\alpha^{-1} & \alpha^2+\alpha^{-1}
		\end{array}
		\right] &
		\left[ \begin{array}{cccc}
			\alpha+1 & 1 & \alpha+1 & \alpha\\
			\alpha & 1 & \alpha+1+\alpha^{-1} & \alpha+1\\
			1 & 1 & \alpha^{-2} & \alpha^{-1}\\
			\alpha+\alpha^{-1} & \alpha^{-1} & \alpha+1+\alpha^{-3} & \alpha+\alpha^{-2}
		\end{array}
		\right ], \\
		(8) \quad\quad\quad\quad\quad\quad\quad\quad\quad & \quad\quad\quad\quad\quad\quad\quad\quad\quad (9)
	\end{split}
\end{equation*}

\begin{equation*}
	\begin{split}
		\left[ \begin{array}{cccc}
			\alpha^2+\alpha^{-1} & \alpha^{-1} & \alpha^3+\alpha & \alpha^3\\
			\alpha^2+\alpha+\alpha^{-1} & \alpha^{-1} & \alpha^3+\alpha^2+\alpha & \alpha^3+\alpha^2\\
			1 & 1 & 1 & 1\\
			\alpha+1 & 1 & \alpha^2 & \alpha^2+1
		\end{array}
		\right ] &
		\left[ \begin{array}{cccc}
			\alpha+\alpha^{-2} & \alpha^{-1} & \alpha+1 & \alpha\\
			\alpha+1+\alpha^{-3} & \alpha^{-2} & \alpha+1+\alpha^{-1} & \alpha+1\\
			\alpha^{-1} & 1 & 1 & 1\\
			\alpha+\alpha^{-1} & 1 & \alpha & \alpha+1
		\end{array}
		\right ]. \\
		(10) \quad\quad\quad\quad\quad\quad\quad\quad\quad & \quad\quad\quad\quad\quad\quad\quad\quad\quad (12)    
	\end{split}
\end{equation*}

\section{Construction of higher order MDS matrices}

Due to the high efficiency of our algorithm, we are able to search for higher-order MDS matrices. Below, we present the implementation tree structures for MDS matrices of order 5 and 6, along with the actual matrices. Specifically, we demonstrate a 5th-order MDS matrix with a word length of 8 that can be implemented using 114 xor operations, and a 6th-order MDS matrix with a word length of 8 that can be implemented using 148 xor operations.

\subsection{5th-order MDS Matrix}

Using Algorithm \ref{alg:search-simplest-trees}, we can obtain the implementation tree of a 5th-order MDS matrix. Figure \ref{fig:4-3} illustrates such an implementation tree for a $5\times 5$ MDS matrix.

\begin{figure}[H]
	\centering
	\includegraphics[height=6cm, width=10cm]{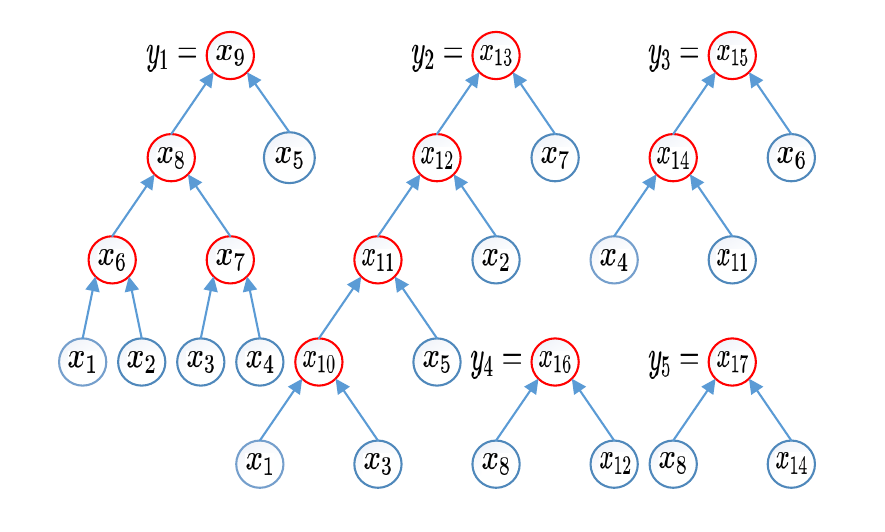}
	\caption{An implementation tree for a $5\times 5$ MDS matrix.}\label{fig:4-3}
\end{figure}

When dealing with matrices of higher orders, we arrive the case that the number of parameters to be assigned increases significantly, so we necessitate a filtering process. Consider the implementation tree of the $5\times 5$ MDS matrix depicted in Figure \ref{fig:4-3} as an example. This tree is comprised of 12 XOR operations, and totally 24 parameters ($a_1, a_2, \ldots, a_{24}$) need to be assigned, where each $a_i$ represents an $n\times n$ matrix over $F_2$ and $n$ denotes the word length:
$$ x_6 = a_1\cdot x_1 \oplus a_2\cdot x_2,\ x_7 = a_3\cdot x_3 \oplus a_4\cdot x_4,\ x_8 = a_5\cdot x_6 \oplus a_6\cdot x_7,\ x_9 = a_7\cdot x_5 \oplus a_8\cdot x_8 = y_1, $$
$$ x_{10} = a_9\cdot x_1 \oplus a_{10}\cdot x_3,\ x_{11} = a_{11}\cdot x_5 \oplus a_{12}\cdot x_8,\ x_{12} = a_{13}\cdot x_2 \oplus a_{14}\cdot x_{11},$$
$$ x_{13} = a_{15}\cdot x_7 \oplus a_{16}\cdot x_{12} = y_2 ,$$
$$ x_{14} =a_{17}\cdot x_4 \oplus a_{18}\cdot x_{11},\ x_{15} = a_{19}\cdot x_6 \oplus a_{20}\cdot x_{14} = y_3, $$
$$ x_{16} = a_{21}\cdot x_8 \oplus a_{22}\cdot x_{12} = y_4, $$
$$ x_{17} = a_{23}\cdot x_8 \oplus a_{24}\cdot x_{14} = y_5, $$

To minimize the implementation cost of the MDS matrix, we aim to maximize the number of identity matrices among the assigned parameters ($a_1, a_2, \ldots, a_{24}$). This simplification process involves searching for the minimum number of non-identity matrices required to form an MDS matrix implementation tree. Initially, we assume that at least $s$ non-identity matrices are needed. We start with $s = 1$ and iterate through all $C_{24}^s$ possible combinations of positions for these non-identity matrices. If a valid MDS matrix can be formed with this assignment, we have a simplified implementation tree. Otherwise, we increment $s$ and repeat the process until a valid tree is found.

In practice, it is found that an implementation tree for a $5\times 5$ MDS matrix can be constructed using only 5 non-identity matrices among the 24 positions. Such a tree is given by the following equations:
\begin{align*}
	x_6 &= x_1 \oplus b_1 \cdot x_2, & x_7 &= x_3 \oplus x_4, \\
	x_8 &= b_2 \cdot x_6 \oplus b_3 \cdot x_7, & x_9 &= x_5 \oplus x_8 = y_1, \\
	x_{10} &= x_1 \oplus x_3, & x_{11} &= x_5 \oplus x_8, \\
	x_{12} &= x_2 \oplus b_4 \cdot x_{11}, & x_{13} &= x_7 \oplus x_{12} = y_2, \\
	x_{14} &= x_4 \oplus b_5 \cdot x_{11}, & x_{15} &= x_6 \oplus x_{14} = y_3, \\
	x_{16} &= x_8 \oplus x_{12} = y_4, & x_{17} &= x_8 \oplus x_{14} = y_5.
\end{align*}

The algorithm for simplifying the MDS matrix implementation tree structure is outlined in Algorithm \ref{alg:simplify-mds-tree}.

\begin{algorithm}[H]
	\caption{Simplify the MDS Matrix Implementation Tree Structure}\label{alg:simplify-mds-tree}
	\KwIn{The $k\times k$ MDS matrix $M_k$ on $R_n$, the implementation tree $T_k$ of $M_k$, and its type $(a_1, a_2, \ldots, a_k)$.}
	\KwOut{The simplified implementation tree $T_k'$ of $M_k$.}
	\For{$s \gets 1$ \KwTo $2i_k$}{
		\ForEach{combination of $s$ positions from $2i_k$}{
			Assign non-identity matrices to these positions and test if the resulting tree forms an MDS matrix\;
			\If{a valid MDS matrix is formed}{
				Output the positions of the $s$ non-identity matrix parameters\;
				\Return the simplified implementation tree\;
			}
		}
	}
\end{algorithm}

From the preceding example, it is evident that Algorithm \ref{alg:simplify-mds-tree} significantly reduces the number of parameters that require assignment. For these reduced parameters, we can directly assign values. Consider the implementation tree of the $5\times 5$ MDS matrix presented in the previous section (Figure \ref{fig:4-3}) as a case study. Assuming a word length of 8, each parameter takes an $8\times 8$ matrix over $F_2$.

Initially, we select the matrix $\alpha$ as our basis, whose minimum polynomial is $x^8+x^6+x^5+x^3+1$:
\begin{small}
	\begin{equation*}
		\alpha =
		\begin{bmatrix}
			0 & 1 & 0 & 0 & 0 & 0 & 0 & 1 \\
			1 & 0 & 0 & 0 & 0 & 0 & 0 & 0 \\
			0 & 1 & 0 & 0 & 1 & 0 & 0 & 0 \\
			0 & 0 & 1 & 0 & 0 & 0 & 0 & 0 \\
			0 & 0 & 0 & 1 & 0 & 0 & 0 & 0 \\
			0 & 0 & 0 & 0 & 1 & 0 & 0 & 0 \\
			0 & 0 & 0 & 0 & 0 & 1 & 0 & 0 \\
			0 & 0 & 0 & 0 & 0 & 0 & 1 & 0 \\
		\end{bmatrix}.
	\end{equation*}
\end{small}
This polynomial is primitive of degree 8 and requires only 2 xor operations for implementation. To minimize the cost of matrix implementation, we opt for parameter values from lower powers of $\alpha$, such as $\alpha^{-4},\alpha^{-3},\alpha^{-2},\alpha^{-1},\alpha^{1},\alpha^{2},\alpha^{3},\alpha^{4}$. For a given $5\times 5$ MDS matrix, we need to go through $8^5 = 2^{15}$ possible assignments.

\begin{algorithm}[]
	\caption{Assign values to the implementation tree of an MDS matrix to obtain a specific MDS matrix}
	\label{alg:4}
	\LinesNumbered
	\KwIn{The $k\times k$ MDS matrix $M_k$ on $R_n$, the simplified implementation tree $T_k$ of $M_k$, positions $b_1,\ldots,b_t$ for parameter assignment, substrate matrix $\alpha$, and parameter assignment range $\{\alpha^{-L},\ldots,\alpha^L\}$.}
	\KwOut{A concrete implementation of the $k\times k$ MDS matrix $M_k$ on $R_n$.}
	\ForEach{$(b_1,\ldots,b_t) \in \{\alpha^{-L},\ldots,\alpha^L\}^t$}{
		Test if the parameter values form an MDS matrix.\\
		\If{The values of $b_1,\ldots,b_t$ form an MDS matrix}{
			Output the values of $b_1,\ldots,b_t$.
		}
	}
	\If{None of the assignments form an MDS matrix}{
		Output "Algorithm failed."
	}
\end{algorithm}

Running Algorithms \ref{alg:simplify-mds-tree} and \ref{alg:4}, we get a concrete implementation of the $5\times 5$ MDS matrix on $GL(2,8)$:
\begin{align*}
	x_6 &= x_1 \oplus \alpha \cdot x_2, & x_7 &= x_3 \oplus x_4, \\
	x_8 &= \alpha^{-2} \cdot x_6 \oplus \alpha^2 \cdot x_7, & x_9 &= x_5 \oplus x_8 = y_1, \\
	x_{10} &= x_1 \oplus x_3, & x_{11} &= x_5 \oplus x_8, \\
	x_{12} &= x_2 \oplus \alpha^{-3} \cdot x_{11}, & x_{13} &= x_7 \oplus x_{12} = y_2, \\
	x_{14} &= x_4 \oplus \alpha \cdot x_{11}, & x_{15} &= x_6 \oplus x_{14} = y_3, \\
	x_{16} &= x_8 \oplus x_{12} = y_4, & x_{17} &= x_8 \oplus x_{14} = y_5,
\end{align*}
where $\alpha$ is as defined earlier, with a minimal polynomial of $x^8+x^6+x^5+x^3+1$. The overall implementation cost is $12 \times 8 + 9 \times 2 = 114$ xor gates.

\subsection{6th-order MDS Matrix}

Similarly to the 5th-order matrix, we do the same for a 6th-order matrix. An implementation tree of 6th-order MDS matrix is depicted below:
\begin{figure}[H]
	\centering
	\includegraphics[height=7cm, width=10cm]{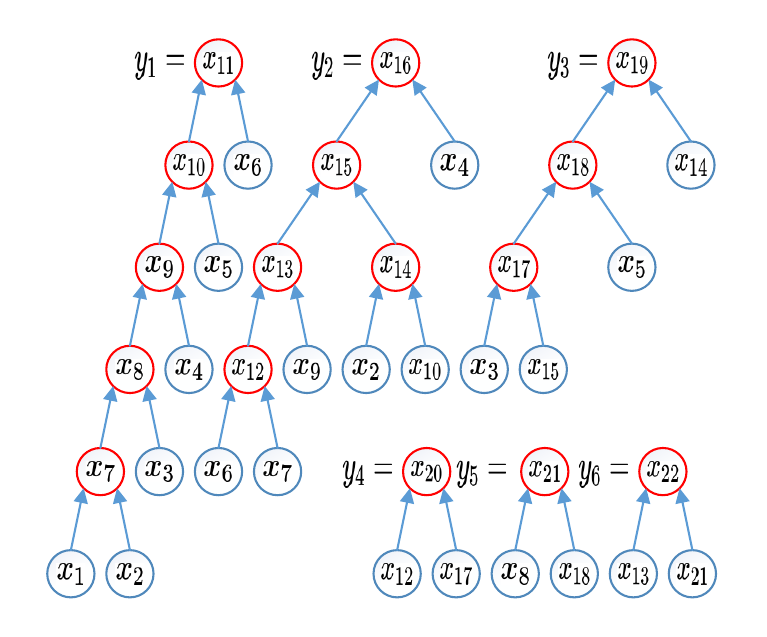}
	\caption{An implementation tree for a $6\times 6$ MDS matrix}\label{fig:4-4}
\end{figure}
Its corresponding matrix is:
$$ x_7 =  x_1 \oplus  x_2,\ x_8 =  x_3 \oplus \alpha\cdot x_7,\ x_9 =  x_4 \oplus \alpha^3\cdot x_8,\ x_{10} =  x_5 \oplus  \alpha^{-1}\cdot x_9,\ x_{11} =  x_6 \oplus  x_{10} = y_1, $$
$$ x_{12} =  x_6 \oplus  x_7,\ x_{13} =  x_9 \oplus  x_{12},\ x_{14} =  x_2 \oplus \alpha^{-1}\cdot x_{10},\ x_{15} =  \alpha^{2}\cdot x_{13} \oplus  x_{14},$$
$$ x_{16} =  x_4 \oplus  x_{15} = y_2, $$
$$ x_{17} = x_3 \oplus x_{15},\ x_{18} =  x_5 \oplus \alpha^2\cdot x_{17},\ x_{19} =  x_{14} \oplus  x_{18} = y_3, $$
$$ x_{20} =  x_{12} \oplus  x_{17} = y_4, $$
$$ x_{21} =  x_8 \oplus  x_{18} = y_5, $$
$$ x_{22} =  x_{13} \oplus  x_{21} = y_6,$$
where
\begin{small}
	\begin{equation*}
		\centering
		\alpha=\left[ \begin{array}{cccccccc}
			0 & 1 & 0 & 0 & 0 & 0 & 0 & 1\\
			1 & 0 & 0 & 0 & 0 & 0 & 0 & 0\\
			0 & 1 & 0 & 0 & 1 & 0 & 0 & 0\\
			0 & 0 & 1 & 0 & 0 & 0 & 0 & 0\\
			0 & 0 & 0 & 1 & 0 & 0 & 0 & 0\\
			0 & 0 & 0 & 0 & 1 & 0 & 0 & 0\\
			0 & 0 & 0 & 0 & 0 & 1 & 0 & 0\\
			0 & 0 & 0 & 0 & 0 & 0 & 1 & 0\\
		\end{array}
		\right]
	\end{equation*}
\end{small}
and its minimal polynomial is $\alpha$ is $x^8+x^6+x^5+x^3+1$. The cost of the entire implementation is $16\times 8 + 10\times 2 = 148$ xor gates.

\section{Conlusion}

In this paper we present a traversal algorithm tailored for the discovery of lightweight $4\times 4$ MDS matrices. As results, we successfully generate all implementation trees for $4\times 4$ MDS matrices using only 8 word-XOR operations. Based on this systematic approach, we derive a series of MDS matrices that exhibit the lowest computational cost. Additionally, we obtain the lowest-cost $4\times 4$ involutory MDS matrix currently known over the finite field $GL(2,8)$. Looking beyond the scope of $4\times 4$ matrices, we extend our method to higher-order MDS matrices. Specifically, we employ our method to obtain the lowest-cost $5\times 5$ and $6\times 6$ MDS matrices as well.

\bibliography{ref}

\begin{appendices}
	
	\section{The simplest trees with capacity 8}
	
	Tree 1-5 are have the type of (3, 3, 1, 1) and tree 6-8 have the type (4, 2, 1, 1).
	
	\begin{figure}[H]
		\centering
		\includegraphics[height=5.5cm, width=9.06cm]{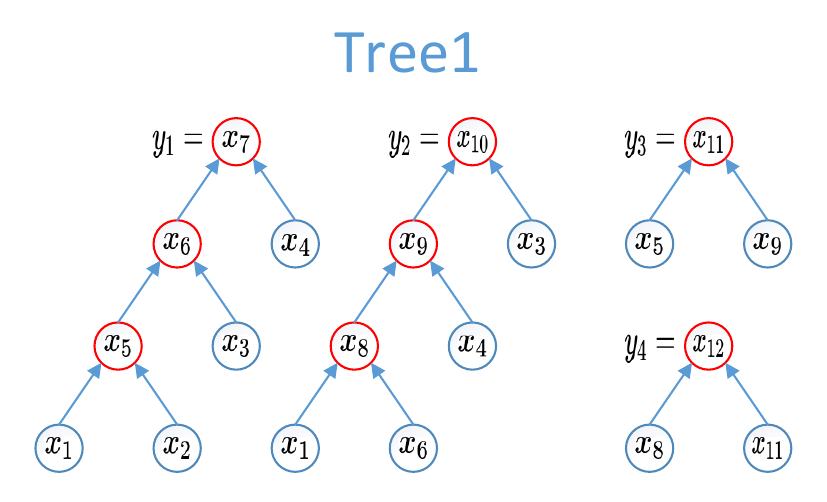}
	\end{figure}
	
	\begin{figure}[H]
		\centering
		\includegraphics[height=5.5cm, width=9.06cm]{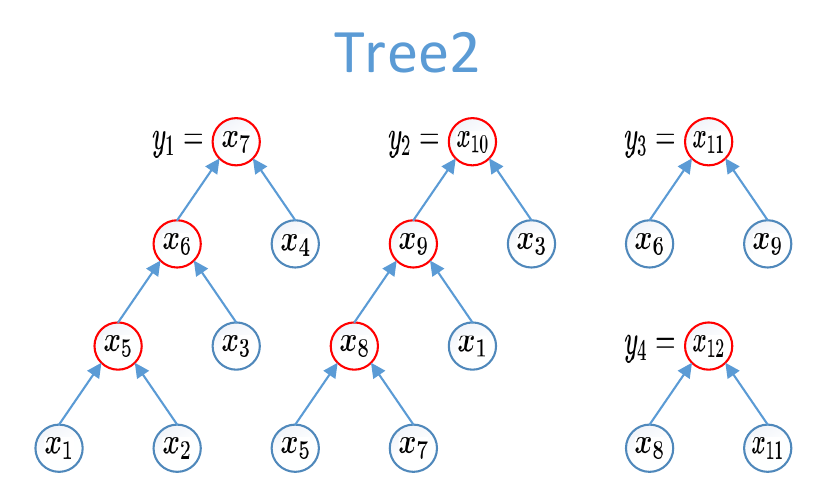}
	\end{figure}
	
	\begin{figure}[H]
		\centering
		\includegraphics[height=5.5cm, width=9.06cm]{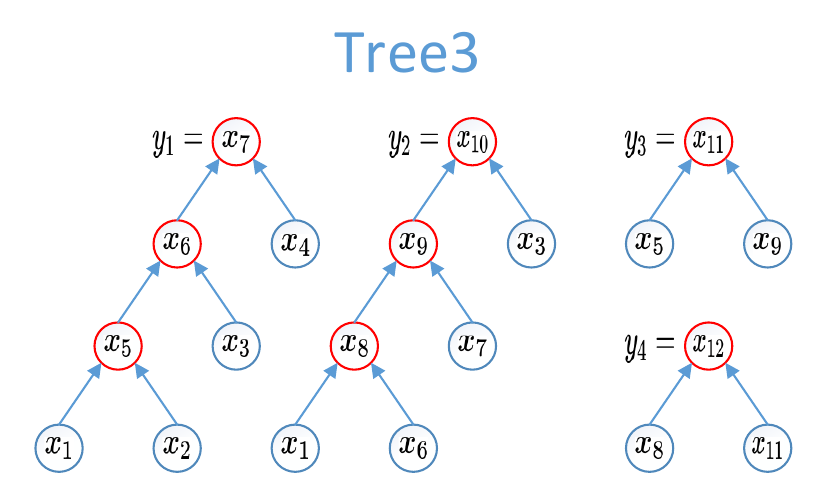}
	\end{figure}
	
	\begin{figure}[H]
		\centering
		\includegraphics[height=5.5cm, width=9.06cm]{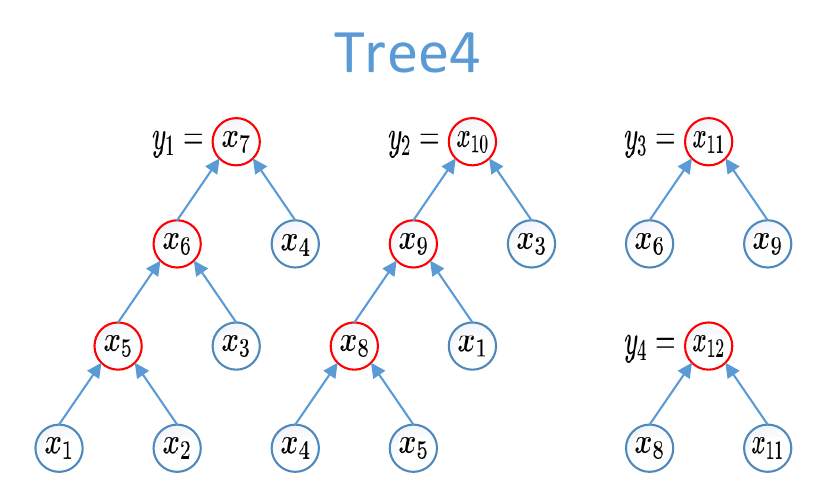}
	\end{figure}
	
	\begin{figure}[H]
		\centering
		\includegraphics[height=6cm, width=8.82cm]{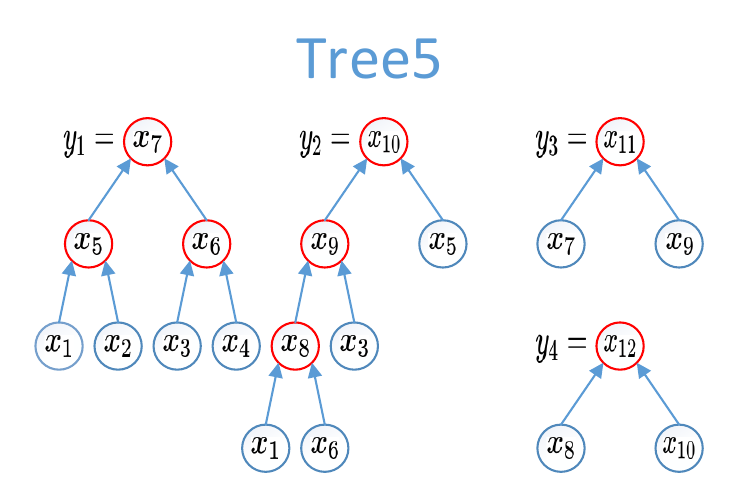}
	\end{figure}
	
	\begin{figure}[H]
		\centering
		\includegraphics[height=5.5cm, width=8.09cm]{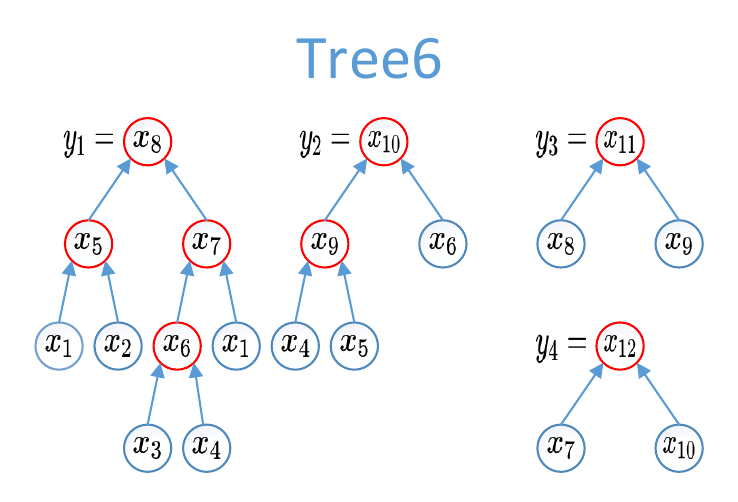}
	\end{figure}
	
	\begin{figure}[H]
		\centering
		\includegraphics[height=5.5cm, width=8.09cm]{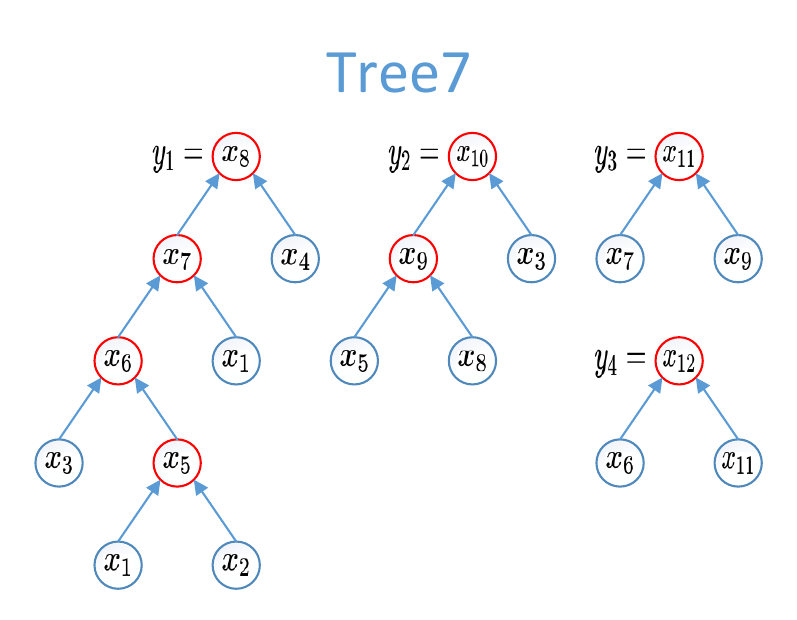}
	\end{figure}
	
	\begin{figure}[H]
		\centering
		\includegraphics[height=6.5cm, width=8.36cm]{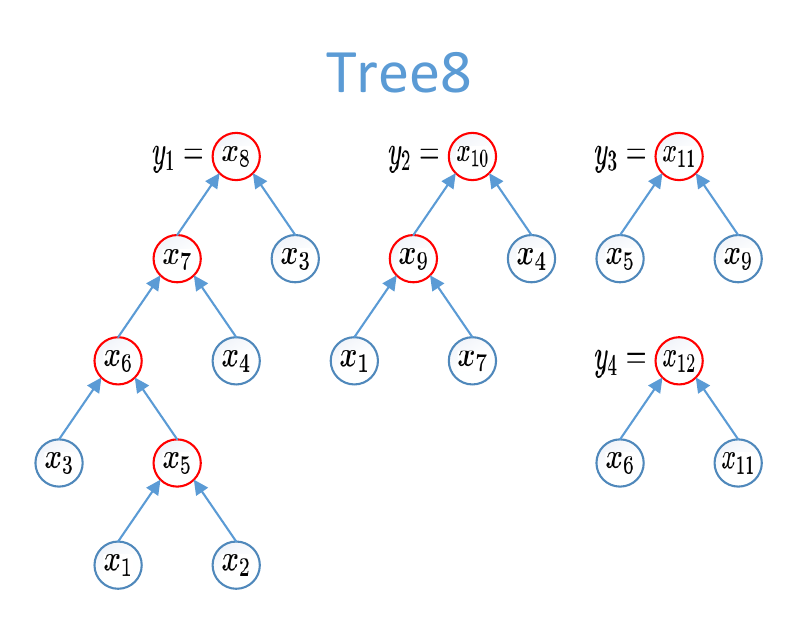}
	\end{figure}
	
	\section{The $4\times 4$ MDS matrices for different depths}

	\begin{table*}[!ht]
		\scriptsize
		\centering
		\caption{The $4\times 4$ MDS matrices for depth 4}
		\label{tab:depth4}
		\resizebox{\textwidth}{3.3cm}{
			\begin{tabular}{c|c|c}
				\hline
				Number & Matrix & Implimentation\\
				\hline
				\multirow{4}{*}{1}&\multirow{4}{*}{\scalebox{0.8}{ $\begin{array}{cccc} \alpha+\alpha^2 & \alpha^2 & 1 & 1\\
						1 & 1 & \alpha & 1+\alpha\\
						\alpha^2 & \alpha+\alpha^2 & 1 & 1+\alpha\\
						1+\alpha & 1 & 1+\alpha & \alpha
					\end{array}
					$}}
				& $x_5 =  x_1 \oplus  x_2,\ x_6 =  x_3 \oplus  x_4,\ x_7 = \alpha\cdot x_1 \oplus  x_6,\  x_8 = \alpha^2\cdot x_5 \oplus  x_7 = y_1$\\
				& & $x_{9} =    x_4 \oplus  x_5,\  x_{10} = \alpha\cdot x_6 \oplus  x_{9} = y_2$\\
				& & $x_{11} =   x_8 \oplus \alpha\cdot x_{9} = y_3$\\
				& & $x_{12} =   x_7 \oplus  x_{10} = y_4$\\
				\hline
				\multirow{4}{*}{2}&\multirow{4}{*}{\scalebox{0.8}{ $\begin{array}{cccc} \alpha^{-2}+\alpha^{-1} & \alpha^{-2} & 1 & 1\\
						1 & 1 & \alpha^{-1} & \alpha^{-1}+1\\
						\alpha^{-2} & \alpha^{-2}+\alpha^{-1} & 1 & \alpha^{-1}+1\\
						\alpha^{-1}+1 & 1 & \alpha^{-1}+1 & \alpha^{-1}
					\end{array}
					$} }
				& $x_5 =  x_1 \oplus  x_2,\ x_6 =  x_3 \oplus  x_4,\ x_7 = \alpha^{-1}\cdot x_1 \oplus  x_6,\  x_8 = \alpha^{-2}\cdot x_5 \oplus  x_7 = y_1$\\
				& & $x_{9} =    x_4 \oplus  x_5,\  x_{10} = \alpha^{-1}\cdot x_6 \oplus  x_{9} = y_2$\\
				& & $x_{11} =   x_8 \oplus \alpha^{-1}\cdot x_{9} = y_3$\\
				& & $x_{12} =   x_7 \oplus  x_{10} = y_4$\\
				\hline
				\multirow{4}{*}{3}&\multirow{4}{*}{\scalebox{0.8}{ $\begin{array}{cccc} \alpha^{-1}+\alpha^2 & \alpha^2 & 1 & 1\\
						1 & 1 & \alpha^{-1} & \alpha^{-1}+\alpha\\
						\alpha^{-1}+1+\alpha^2 & 1+\alpha^2 & 1 & 1+\alpha\\
						\alpha^{-1}+1 & 1 & \alpha^{-1}+1 & \alpha^{-1}+1+\alpha
					\end{array}
					$} }
				& $x_5 =  x_1 \oplus  x_2,\ x_6 =  x_3 \oplus  x_4,\ x_7 = \alpha^{-1}\cdot x_1 \oplus  x_6,\  x_8 = \alpha^2\cdot x_5 \oplus  x_7 = y_1$\\
				& & $x_{9} =   \alpha\cdot x_4 \oplus  x_5,\  x_{10} = \alpha^{-1}\cdot x_6 \oplus  x_{9} = y_2$\\
				& & $x_{11} =   x_8 \oplus  x_{9} = y_3$\\
				& & $x_{12} =   x_7 \oplus  x_{10} = y_4$\\
				\hline
				\multirow{4}{*}{4}&\multirow{4}{*}{\scalebox{0.8}{ $\begin{array}{cccc} \alpha^{-1}+\alpha^2 & \alpha^2 & 1 & 1\\
						1 & 1 & \alpha^{-1} & \alpha^{-1}+1\\
						\alpha^2 & \alpha^{-1}+\alpha^2 & 1 & \alpha^{-1}+1\\
						\alpha^{-1}+1 & 1 & \alpha^{-1}+1 & \alpha^{-1}
					\end{array}
					$} }
				& $x_5 =  x_1 \oplus  x_2,\ x_6 =  x_3 \oplus  x_4,\ x_7 = \alpha^{-1}\cdot x_1 \oplus  x_6,\  x_8 = \alpha^2\cdot x_5 \oplus  x_7 = y_1$\\
				& & $x_{9} =    x_4 \oplus  x_5,\  x_{10} = \alpha^{-1}\cdot x_6 \oplus  x_{9} = y_2$\\
				& & $x_{11} =   x_8 \oplus \alpha^{-1}\cdot x_{9} = y_3$\\
				& & $x_{12} =   x_7 \oplus  x_{10} = y_4$\\
				\hline
				\multirow{4}{*}{5}&\multirow{4}{*}{\scalebox{0.8}{ $\begin{array}{cccc} \alpha^{-2}+1 & \alpha^{-2} & 1 & 1\\
						1 & 1 & \alpha^{-1} & \alpha^{-1}+\alpha\\
						\alpha^{-2} & \alpha^{-2}+1 & 1 & 1+\alpha\\
						1+\alpha & 1 & \alpha^{-1}+\alpha & \alpha^{-1}
					\end{array}
					$} }
				& $x_5 =  x_1 \oplus  x_2,\ x_6 =  x_3 \oplus  x_4,\ x_7 =  x_1 \oplus  x_6,\  x_8 = \alpha^{-2}\cdot x_5 \oplus  x_7 = y_1$\\
				& & $x_{9} =   \alpha\cdot x_4 \oplus  x_5,\  x_{10} = \alpha^{-1}\cdot x_6 \oplus  x_{9} = y_2$\\
				& & $x_{11} =   x_8 \oplus  x_{9} = y_3$\\
				& & $x_{12} =  \alpha\cdot x_7 \oplus  x_{10} = y_4$\\
				\hline
				\multirow{4}{*}{6}&\multirow{4}{*}{\scalebox{0.8}{ $\begin{array}{cccc} \alpha^{-2}+1 & \alpha^{-2} & 1 & 1\\
						1 & 1 & \alpha & \alpha^{-1}+\alpha\\
						\alpha^{-2} & \alpha^{-2}+1 & 1 & \alpha^{-1}+1\\
						\alpha^{-1}+1 & 1 & \alpha^{-1}+\alpha & \alpha
					\end{array}
					$} }
				& $x_5 =  x_1 \oplus  x_2,\ x_6 =  x_3 \oplus  x_4,\ x_7 =  x_1 \oplus  x_6,\  x_8 = \alpha^{-2}\cdot x_5 \oplus  x_7 = y_1$\\
				& & $x_{9} =   \alpha^{-1}\cdot x_4 \oplus  x_5,\  x_{10} = \alpha\cdot x_6 \oplus  x_{9} = y_2$\\
				& & $x_{11} =   x_8 \oplus  x_{9} = y_3$\\
				& & $x_{12} =  \alpha^{-1}\cdot x_7 \oplus  x_{10} = y_4$\\
				\hline
				\multirow{4}{*}{7}&\multirow{4}{*}{\scalebox{0.8}{ $\begin{array}{cccc} 1+\alpha^2 & \alpha^2 & 1 & 1\\
						1 & 1 & \alpha^{-1} & \alpha^{-1}+1\\
						\alpha^{-1}+1+\alpha^2 & \alpha^{-1}+\alpha^2 & 1 & \alpha^{-1}+1\\
						1+\alpha & 1 & \alpha^{-1}+\alpha & \alpha^{-1}+1+\alpha
					\end{array}
					$} }
				& $x_5 =  x_1 \oplus  x_2,\ x_6 =  x_3 \oplus  x_4,\ x_7 =  x_1 \oplus  x_6,\  x_8 = \alpha^2\cdot x_5 \oplus  x_7 = y_1$\\
				& & $x_{9} =    x_4 \oplus  x_5,\  x_{10} = \alpha^{-1}\cdot x_6 \oplus  x_{9} = y_2$\\
				& & $x_{11} =   x_8 \oplus \alpha^{-1}\cdot x_{9} = y_3$\\
				& & $x_{12} =  \alpha\cdot x_7 \oplus  x_{10} = y_4$\\
				\hline
			\end{tabular}
		}
	\end{table*}
	
	In Table \ref{tab:depth4}, $\alpha$ represents the companion matrix of either $x^8+x^2+1$ or. The $4\times 4$ MDS matrices presented in Table \ref{tab:depth4} for depth 4 can be implemented using 69 XOR operations (for 8-bit input) or 37 XOR operations (for 4-bit input), representing the most cost-effective solution at this depth.
	
	As the depth decreases to 3, the minimum cost increases to 77 XOR operations (for 8-bit input) or 41 XOR operations (for 4-bit input). Owing to the strict limitation on depth, the available parameter options become constrained. For depth 3, there are three matrices that achieve the lowest cost. These matrices are enumerated below.
	\begin{equation*}
		\centering
		1:\left[ \begin{array}{cccc}
			1 & 1 & \alpha^{-1} & \alpha^{-1}+\alpha\\
			1 & 1+\alpha & \alpha & \alpha\\
			\alpha & \alpha^{-1} & 1+\alpha^{-1} & 1\\
			1+\alpha & 1 & 1 & 1+\alpha
		\end{array}
		\right ]
	\end{equation*}
	
	$$ x_5 =  x_1 \oplus  x_2,\ x_6 = \alpha\cdot x_4 \oplus  x_5,\ x_7 =  x_3 \oplus  x_4,\  x_8 =  x_6 \oplus \alpha^{-1}\cdot x_7 = y_1. $$
	$$ x_{9} =    x_2 \oplus  x_3,\  x_{10} = x_6 \oplus \alpha\cdot x_{9} = y_2. $$
	$$ x_{11} =  \alpha\cdot x_1 \oplus  x_{7},\ x_{12} = \alpha^{-1}\cdot x_9 \oplus  x_{11} = y_3. $$
	$$ x_{13} =   x_6 \oplus  x_{11} = y_4. $$
	
	\begin{equation*}
		\centering
		2:\left[ \begin{array}{cccc}
			1 & 1 & \alpha & \alpha^{-1}+\alpha\\
			1 & 1+\alpha^{-1} & \alpha^{-1} & \alpha^{-1}\\
			\alpha^{-1} & \alpha & 1+\alpha & 1\\
			1+\alpha^{-1} & 1 & 1 & 1+\alpha^{-1}
		\end{array}
		\right ]
	\end{equation*}
	
	$$ x_5 =  x_1 \oplus  x_2,\ x_6 = \alpha^{-1}\cdot x_4 \oplus  x_5,\ x_7 =  x_3 \oplus  x_4,\  x_8 =  x_6 \oplus \alpha\cdot x_7 = y_1 $$
	$$ x_{9} =    x_2 \oplus  x_3,\  x_{10} = x_6 \oplus \alpha^{-1}\cdot x_{9} = y_2 $$
	$$ x_{11} =  \alpha^{-1}\cdot x_1 \oplus  x_{7},\ x_{12} = \alpha\cdot x_9 \oplus  x_{11} = y_3 $$
	$$ x_{13} =   x_6 \oplus  x_{11} = y_4 $$
	
	\begin{equation*}
		\centering
		3:\left[ \begin{array}{cccc}
			\alpha^{-1} & \alpha & 1 & \alpha^{-1}\\
			\alpha^{-1}+1 & 1+\alpha & 1 & 1\\
			1& 1 & \alpha^{-1}+1 & \alpha^{-2}+1\\
			1+\alpha & 1 & 1+\alpha & 1
		\end{array}
		\right ]
	\end{equation*}
	
	$$ x_5 = \alpha^{-1}\cdot x_1 \oplus \alpha\cdot x_2,\ x_6 =  x_3 \oplus \alpha^{-1}\cdot x_4,\ x_7 =  x_5 \oplus  x_6=y_1   $$
	$$ x_{8} =    x_1 \oplus  x_3,\  x_{9} = x_2 \oplus \alpha^{-1}\cdot x_{4},\ x_{10} =  x_8 \oplus  x_9,\ x_{11} =  x_5 \oplus  x_{10} = y_2 $$
	$$ x_{12} = \alpha^{-1}\cdot x_6 \oplus  x_{10} = y_3 $$
	$$ x_{13} =  \alpha\cdot x_8 \oplus  x_{11} = y_4 $$
\end{appendices}

\end{document}